\documentclass[11pt,a4paper]{amsart}
\usepackage{amsfonts, amsthm}
\usepackage{hyperref}
\usepackage[english]{babel}
\usepackage{amsmath}
\usepackage{amssymb}
\usepackage{url}
\usepackage{amsxtra}
\usepackage{mathrsfs}
\usepackage{fancyhdr}
\usepackage{enumerate}
\usepackage{graphicx}
\usepackage{pdfsync}
\usepackage{psfrag}
\usepackage{tikz}
\usepackage{soul}

\tikzstyle{nodino}=[circle,draw,fill,inner sep=0pt,minimum size=0.5mm]
\tikzstyle{infinito}=[circle,inner sep=0pt,minimum size=0mm]
\tikzstyle{nodo}=[circle,draw,fill,inner sep=0pt,minimum size=0.5*\widthof{k}]

\newcommand{\blu}[1]{{\color{blue}#1}}

\newtheorem{theorem}{Theorem}

\newtheorem{lemma}{Lemma}[section]
\newtheorem{proposition}[lemma]{Proposition}

\newtheorem{assumption}{Assumption}

\theoremstyle{definition}
\newtheorem{remark}[lemma]{Remark}

\numberwithin{equation}{section}

\newcommand{\beq}{\begin{equation}}
\newcommand{\eeq}{\end{equation}}
\newcommand{\be}{\begin{equation*}}
\newcommand{\ee}{\end{equation*}}
\newcommand{\n}{\noindent}

\newcommand{\vertiii}[1]{{\left\vert\kern-0.25ex\left\vert\kern-0.25ex\left\vert #1 
    \right\vert\kern-0.25ex\right\vert\kern-0.25ex\right\vert}}

\newcommand{\RE}{\mathbb R}
\newcommand{\erre}{\mathbb R}

\newcommand{\NA}{\mathbb N}

\newcommand{\DD}{\mathcal D}

\newcommand{\GG}{\mathcal{G}}

\newcommand{\supp}{\operatorname{supp}\,}

\newcommand{\lf}{\left}
\newcommand{\ri}{\right}
\newcommand{\ve}{\varepsilon}
\newcommand{\al}{\alpha}

\newcommand{\ga}{\gamma}

\newcommand{\la}{\lambda}
\newcommand{\de}{\delta}

\newcommand{\ome}{\omega}
\DeclareMathOperator{\sech}{sech}

\renewcommand{\Re}{\operatorname{Re}\,}

\renewcommand{\leq}{\leqslant}
\renewcommand{\geq}{\geqslant}

\newcommand{\x}{\underline{x}}
\newcommand{\y}{\underline{y}}
\renewcommand{\v}{\underline{v}}

\newcommand{\f}{\frac}
\newcommand{\EE}{\mathcal E}
\newcommand{\C}{\mathbb{C}}
\newcommand{\VV}{R}
\newcommand{\WW}{S}
\newcommand{\ZZ}{Z}

\newcommand{\vol}{\operatorname{Vol}}
\newcommand{\one}{{\bf 1}}
\renewcommand{\a}{\underline{a}}
\renewcommand{\b}{\underline{b}}

\title[]{Ground state and orbital stability for the NLS equation on a general starlike graph with potentials}

\author[]{Claudio Cacciapuoti}
\address{Dipartimento di Scienza e Alta Tecnologia, Universit\`a dell'Insubria, Via Valleggio 11, 22100 Como, Italy, EU}
\email{claudio.cacciapuoti@uninsubria.it	}%

\author[]{Domenico Finco}
\address{Facolt\`a di Ingegneria, Universit\`a Telematica
Internazionale Uninettuno,  Corso Vittorio Emanuele II 39, 00186 Roma,
Italy}
\email{d.finco@uninettunouniversity.net}

\author[]{Diego Noja}
\address{Dipartimento di Matematica e Applicazioni, Universit\`a
 di Milano Bicocca,  via R. Cozzi, 53, 20125 Milano, Italy}
\email{diego.noja@unimib.it} 

\date{}

\begin{document}

\begin{abstract} We consider a  nonlinear Schr\"odinger equation (NLS) posed on a graph or network composed of a generic compact part to which a finite number of half-lines are attached. We call this structure a starlike graph. At the vertices of the graph 
interactions of $\delta$-type can be present and an overall external potential is admitted. Under general assumptions on the potential, we prove that the NLS  is globally well-posed in the energy domain.

We are interested in minimizing the energy of the system on the manifold of constant mass ($L^2$-norm). When existing, the minimizer
is called ground state and it is the profile of an orbitally stable standing wave for the NLS evolution. We prove that a ground state exists for sufficiently small masses whenever the  quadratic part of the energy admits a simple isolated eigenvalue at the bottom of the spectrum (the linear ground state). This is a wide generalization of a result previously obtained for a star graph with a single vertex.
The main part of the proof is devoted to prove the concentration compactness principle for starlike structures; this is non trivial
due to the lack of translation invariance of the domain. Then we show that a minimizing bounded $H^1$ sequence for the constrained NLS energy with external linear potentials is in fact convergent if its mass is small enough. Examples are provided with discussion of hypotheses on the linear part.
\end{abstract}

\maketitle

\begin{footnotesize}
 
 \end{footnotesize}

\section{Introduction}
Analysis on metric graphs and networks is a growing subject with many potential applications of physical and technological character. The interest in these structures, also from a mathematical point of view lies in the fact that they are relatively simple analytically, being essentially one dimensional, but on the other hand they can have in a sense arbitrary complexity due to nontrivial connectivity and topology.\\ A large part of the literature is devoted to linear equations on graphs (see \cite{BerKu, Mugnolo} for an overview of theory and the many applications), with special emphasis on Schr\"odinger equation describing the so called quantum graphs. Recently nonlinear equations have attracted attention, and a certain amount of mathematical work has been done on nonlinear Schr\"odinger equation on quantum graphs, at least in some special situations (see for example \cite{[ACFN1],[ACFN2],[ACFN4],[ACFN3],cfn15,[NPS15],[MP16],AST1,AST2,[PS16]}; a review with references to related physical research is in \cite{Noja14}).
In this paper we settle some issues about the nonlinear Schr\"odinger equation on a quantum graph $\GG$, composed by a compact core to which a finite number of  half-lines are attached (and at least one). We refer to this structure as a {\it starlike graph} (see Fig.1). \\ Our main interest is in showing that the NLS dynamics admits on a starlike graph a ground state under mild and natural hypotheses. We mean as ground state a standing solution of NLS on the graph which minimizes the system energy at a fixed constant mass, i.e. $L^2$-norm. A previous result in a very special case was given in the paper \cite{acfn-aihp}, where a single vertex with $N$ half lines -- a so called {\it star graph} -- with a delta interaction was considered.
Here we extend that result widely generalizing the topology of the compact core, and admitting the (possible) presence of external potentials on the graph. We however retain the power nonlinearity to avoid wordy statements but this limitation is not really necessary. The NLS on the graph is an equation of the form
\begin{equation}\label{equation}
i\frac{d}{dt}\Psi=H\Psi -|\Psi|^{2\mu}\Psi
\end{equation}
where: 
\begin{enumerate}[i)]
\item  $\Psi$ is a multicomponent function where every component is a complex function on a single edge of the graph; 
\item The operator $H$ is a Schr\"odinger operator on the graph acting on every edge as $-\frac{d^2}{dx^2} +W$ and complemented with suitable boundary condition to make it selfadjoint on its domain; 
\item The nonlinearity, of power type and again defined edge by edge, is focusing (the minus sign).
\end{enumerate} 
For further details and complete hypotheses and definitions, see the following section. 
The previous equation is globally well-posed in energy or form domain $H^1(\GG)$, which is the usual Sobolev space including continuity at vertices, for every $\mu\in [0,2)$, the {\it subcritical} range, see Section \ref{s:wp} below for a proof. In the {\it critical} case $\mu=2$ the solution is only defined for small initial data, as in the case of the line. In any case the mass of the solution, i.e. its $L^2$-norm $\|\Psi\|^2$, and the energy
\be 
E[\Psi]=E^{lin}[\Psi ] -\frac{1}{\mu+1} \|\Psi \|_{2\mu+2}^{2\mu + 2} 
= \| \Psi ' \|^2 +(\Psi,W\Psi)+ \sum_{\v\in V} \al(\v) |\Psi(\v)|^2  -\frac{1}{\mu+1} \|\Psi \|_{2\mu+2}^{2\mu + 2} 
\ee
are conserved quantities.
Of special importance is the quadratic contribution to the energy
\be
E^{lin}[\Psi ] = \| \Psi ' \|^2 + (\Psi,W\Psi) +  \sum_{\v\in V} \al(\v) |\Psi(\v)|^2. 
\ee
It contains three terms. The kinetic energy, a potential term defined by $W$ and the last term which is the energy associated to delta interactions concentrated at vertices $\v$ of the graph; we do not assume definite sign on the strengths $\alpha(\v)$ of the interaction at vertices.\\
Our hypotheses are rather simple and they regard only the topology of the graph and the quadratic part of the energy.
\begin{assumption}\label{a:0} $\GG$ is a connected graph with a finite number of edges and vertices, and it is composed by a compact core and at  least one infinite edge (one half-line).
\end{assumption}
\begin{assumption}\label{a:1}
$W = W_+ - W_-$ with $W_\pm\geq 0$, $W_+\in L^1(\GG) + L^{\infty}(\GG)$, and $W_-\in L^r(\GG)$ for some $r\in [1,1+1/\mu]$.  
\end{assumption}
\begin{assumption}\label{a:2} 
$\inf \sigma(H):= -E_0$, $E_0>0$ and it   is an isolated, non degenerate eigenvalue. 
\end{assumption}
Our main theorem gives the existence of nonlinear ground state under the above assumptions.
\begin{theorem} \label{t:prob1}
Let $0<\mu<2 $ and consider on a starlike graph $\GG$ the following minimization problem:
\begin{equation}\label{minimization}
-\nu = \inf \{ E[\Psi] \text{ s.t. } \Psi\in \EE , \, M[\Psi]=m \}.
\end{equation}
If {\text Assumptions} \ref{a:0}, \ref{a:1}, and \ref{a:2} hold true, then $mE_0<\nu<+\infty$ for any $m>0$.  Moreover, there exists $m^\ast>0$ such that for $0< m<m^\ast$  there exists $\hat \Psi \in H^1(\GG)$, with $M[\hat\Psi]=m$,  such that $E[\hat \Psi] =-\nu $.
\end{theorem}
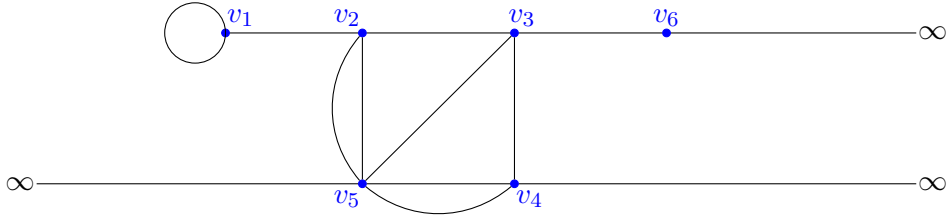
\begin{figure}[t]
\begin{center}
\begin{tikzpicture}
\node at (-6.5,0) [infinito](0) {${\infty}$};
\node at (-2,0) [nodo,blue] (5) {};
\node at (0,0) [nodo,blue] (4) {};
\node at (-2,2) [nodo,blue] (2) {};
\node at (0,2) [nodo,blue] (3) {};
\node at (5.5,0) [infinito] (8) {${\infty}$};
\node at (5.5,2) [infinito] (7) {${\infty}$};
\node at (2,2) [nodo,blue] (6) {};
\node at (-3.8,2) [nodo,blue] (1) {};
\node at (-2.2,2.2){\blu{$v_2$}};
\node at (0.1,2.2){\blu{$v_3$}};
\node at (0.2,-0.2){\blu{$v_4$}};
\node at (-2.2,-0.2){\blu{$v_5$}};
\node at (2,2.2){\blu{$v_6$}};

\node at (-3.6,2.2){\blu{$v_1$}};

\draw [-,black] (-4.2,2) circle (0.4cm) ;
\draw [-,black] (5) -- (2);
\draw [-,black] (0) -- (5);
\draw [-,black] (5) -- (4);
\draw [-,black] (5) to [out=-40,in=-140] (4);
\draw [-,black] (5) to [out=130,in=-130] (2);
\draw [-,black] (5) -- (3);
\draw [-,black] (1) -- (2);
\draw [-,black] (2) -- (3);
\draw [-,black] (3) -- (6);
\draw [-,black] (4) -- (3);
\draw [-,black] (6) -- (7);
\draw [-,black] (4) -- (8);
\end{tikzpicture}
\end{center}
\caption{13 edges (10 interior, 3 exterior); 6 vertices; one tadpole.}
\end{figure}
\noindent
We briefly comment on  the assumptions.\\
Assumption \ref{a:0} is a topological one. We remark that if $\GG$ is a compact connected graph without infinite edges, the minimization problem  \eqref{minimization} admits a solution whenever the energy functional $E[\Psi]$ is bounded from below. \\
Assumption \ref{a:1} is a rather weak hypothesis which is sufficient to guarantee that $E^{lin}$ is the quadratic form of a selfadjoint operator bounded from below, see also Remark \ref{r:train}.  We stress  that  the stronger  assumption  $W_-\in L^r(\GG)$  is needed only in the final part of the proof of Th. \ref{t:prob1}, to guarantee that the $W_-$-terms in the energy functional $E[\Psi]$ are negligible whenever the energy functional is evaluated on sequences that escape at infinity on one of the half-lines (runaway sequences), see Eq. \eqref{holiday} below. All the results before the limit \eqref{holiday} hold true under the weaker assumption  $W\in L^1(\GG) + L^{\infty}(\GG)$.  \\
Assumption \ref{a:2} assures existence of a unique {\it linear} ground state and it is satisfied in many relevant examples, such as the following:
\begin{enumerate}[a)]
\item  No delta terms, i.e. $\alpha(\v) =0$ for all $\v$ (also called Kirchhoff boundary conditions at vertices,  see, e.g.   \cite{[KS99]}) and a sufficiently well behaved and decaying external potential attractive in the mean, i.e. such that $\int_\GG W <0.$ 
In the pure Kirchhoff case (with no potentials) an extensive analysis of NLS with power nonlinearity has been given in the recent papers \cite{AST1,AST2}, where in particular it is shown that existence of a ground state for subcritical nonlinearity holds true only in some exceptional cases, the simplest one being the tadpole graph \cite{cfn15, [NPS15]}. Here we show that summing a small negative potential restores the ground state generically.
\item Absence of potential term and delta interactions negative in the mean: $\sum_{\v\in V} \alpha(\v)<0$ 
(Se also \cite{EJ} for an explicit example in this case).
\item A mixing of the two: delta interaction at the vertices and well behaved potentials with negative potential energy: $\sum_{\v\in V} \alpha(\v) + \int_\GG W<0$.
\end{enumerate}
Notice that at the level of quadratic form and in this one dimensional problem, strictly speaking, one could consider on the same footing both the delta terms and the regular potential term. We have a preference to keep separate the two contributions because this is the usual way they are treated in quantum graph literature.   \\
We comment now briefly on the proof strategy. As in \cite{acfn-aihp} we want to make use of concentration compactness techniques, but we have to cope with the lack of translational invariance of the graph. We show that for starlike graphs the Concentration Compactness Lemma \ref{l:cc} is valid. We note that with respect to the standard concentration compactness result in $\RE^n$ see, e.g. \cite{Caz03, Caz06},  we have to split the compact case in two sub-cases, named {\it runaway} and {\it convergent}. In the runaway case a   minimizing,  bounded in $H^1(\GG)$, sequence $\Psi_n$ eventually escapes on a single distinguished exterior edge, in the sense that any of its  $L^p$-norms with $p\geq 2$ on the other edges vanishes and the same occurs for the $L^p$-norm on any bounded part of the distinguished edge. In the convergence case, which is the one we are interested in, an $H^1(\GG)$-bounded sequence admits a converging subsequence in $L^p(\GG), \ p\geq 2$.
So that, to get convergence,  we have to exclude vanishing, dichotomy and runaway case. In particular, to exclude the runaway case, we prove, by use of Bifurcation Theory, the existence of a branch of nonlinear solutions of the stationary NLS which is born from the linear ground state. This is the point where we make use of Assumption \ref{a:2}. Along this branch the $L^2$-norm of the solution is small near the bifurcation point. We show that  in this case a runaway minimizing sequence has an energy which is not compatible with the energy deduced from bifurcation theory in the small mass regime. Here is the only point where we use the hypothesis of small mass in Th. \ref{t:prob1}. With the present technique it is not possible to exclude that for big masses the minimizing sequence is runaway. In the simpler case of star graph this possibility has been excluded in \cite{ACFN16} by using  a finite dimensional reduction, a procedure which however does not work in more structured graphs or in the presence of external potentials. We remark that in the case of the line with a delta interaction, the existence of the ground state for every value of the mass was given in \cite{ANV12}, which covers also other examples of point interactions, while an even more singular interaction is treated in \cite{ANcmp}.\\
When global well-posedness of the model holds true, the ground state, being a constrained minimum of the energy, is orbitally stable. We provide a global well-posedness result in $H^1(\GG)$ in Th. \ref{t:wp}, filling in  a gap in the literature.\\
We end the introduction with an outline of the paper. In Section 2 we give preliminary definitions and results on quantum graphs (Secs. \ref{s:qg} and \ref{s:gn}), we make precise our hypotheses on the quadratic part of the energy and comment about the validity of Assumption \ref{a:2} (Secs. \ref{s:qf} and \ref{s:lgs}); finally we give well-posedness and mass and energy conservation for the time dependent nonlinear Schr\"odinger equation on a starlike graph (Secs. \ref{s:enlp} and \ref{s:wp}). In Section 3 the Concentration Compactness lemma is extended to the case of starlike networks. All statements are given explicitly, but only the steps which need essential modification of the original result valid on $\RE^N$ are proved while references are provided for the missing but straightforward steps.  
In Section 4 a bifurcation analysis showing the existence of a branch of standing waves emanating from the vanishing solution under the validity of Assumption \ref{a:1} and \ref{a:2} is proved, see Th. \ref{t:bif}. Estimates on the size of the branch element in terms of relevant parameters are given as well. In the last Section 5 we use the  results obtained  in Section 3 and 4 to show that only the convergence case in Concentration Compactness alternative holds if the minimizing sequence has a mass sufficiently small, which ends the construction needed for the proof of Th. \ref{t:prob1}. 

Throughout the paper $c$ and $C$ denote generic positive constants whose value may change form line to line.

\section{Preliminaries}
\subsection{Quantum Graphs\label{s:qg}}
We consider a connected metric graph $\GG=(V,E)$ where $V$ is the set of vertices and $E$ is the set of edges. We assume that the cardinalities $|V|$ and $|E|$ of $V$ and $E$
are finite. We identify each edge $e\in E$  with length $L_e \in (0,\infty] $ with the interval $I_e=[0, L_e]$, if $L_e$ is finite, or $[0,\infty)$, if $L_e$ is infinite.
The set of edges with finite length is denoted by $E^{in}$ while the set of edges with infinite length is denoted by $E^{ex}$.
Moreover we associate each finite length edge with two vertices, and each infinite length edges with one vertex. The notation $\v\in e$ with $\v\in V$ and
$e\in E$, denotes that $\v$ is a vertex of the edge $e$.
Two vertices $\v_1$ and $\v_2$ are adjacent, $\v_1 \sim \v_2$ if they are vertices of a common edge which connects them. The degree of
a vertex is the number of edges emanating from it. We denote by $\{e\prec \v\}$ the set of edges connecting the vertex $e$.
We fix a coordinate $x$ on each interval $I_e$ such that $x=0$ and $x=L_e$ correspond to vertices  if $L_e<\infty$ while if $L_e=\infty$ the vertex
attached to the rest of the graph corresponds to $x=0$.
Any choice of orientation of finite length edges is equivalent for our purposes. 
To avoid ambiguities, from now on we will denote
points on the graph with $\x= (e,x)$, where $e\in E$ identifies the edge and  $x\in I_e $ is the  coordinate on the corresponding edge. 
The length of a path is well defined due to the coordinates on edges and therefore there is a natural distance on $\GG$. Given 
$\x$ and $\y$ on $\GG$ the distance $d(\x, \y) $ is defined as the infimum of the length of the paths connecting the two points.
Then $(\GG, d)$ is a locally compact metric space and it is compact if and only if $L_e <\infty$  for $\forall e \in E$.
In this paper we will assume that there is at least one edge with infinite length, so that the considered graph is non compact.
A function $\Psi: \GG \to \C$ is equivalent to a family of functions $\{ \psi_e \}_{e\in E}$ with $\psi_e : I_e \to \C$. In our notation, if $\x = (e,x)$
\[\Psi(\x) = \psi_e(x).\] 
The spaces $L^p (\GG)$, $1\leq p\leq \infty$, are made of functions $\Psi$ such that $\psi_e \in L^p (I_e)$ for all $ e\in E$ and 
\[
\| \Psi \|_p^p = \sum_{e\in E} \|\psi_e \|^p_{L^p (I_e) }, \, 1\leq p<\infty \qquad
\|\Psi \|_\infty = \max_{e\in E}  \|\psi_e \|_{L^\infty (I_e) }.
\]
We denote by $( \cdot ,\, \cdot )  $ the inner product associated with  $L^2(\GG)$. When $p=2$, the index will be omitted
We denote by $C(\GG)$ the set of continuous functions on $\GG$ and introduce the spaces 
\[H^1(\GG) : = \left\{\Psi\in C(\GG) \;\textrm{s.t.}\; \psi_e \in H^1(I_e) \; \forall e \in E\right\}\]
equipped with the norm 
\[\| \Psi \|_{H^1(\GG)}^2 \ = \ \sum_{e\in E} \| \psi_e \|_{H^1(I_e)}^2.
\]
and 
\[H^2(\GG) : = \left\{\Psi\in H^1(\GG) \;\textrm{s.t.}\; \psi_e \in H^2(I_e) \; \forall e \in E\right\}\]
equipped with the norm 
\be 
\| \Psi \|_{H^2(\GG)}^2 \ = \ \sum_{e\in E} \| \psi_e \|_{H^2(I_e)}^2.
\ee
In the following, whenever a functional norm refers to a function defined on the graph,
 we omit the symbol $\GG$. 
 \subsection{Gagliardo-Nirenberg inequalities on graphs\label{s:gn}} Let $\GG$ be  any non-compact graph, then if    $p,q\in[2, +\infty]$, with $p\geq q$, and $\alpha = \frac{2}{2+q}(1-q/p)$, there exists $C$ such that 
\begin{equation}\label{gajardo2}
\| \Psi \|_{p} \ \leqslant \ C \| \Psi^\prime \|^\alpha \| \Psi \|^{1-\alpha}_q,
\end{equation}
for all $\Psi\in H^1(\GG)$.  \\
A proof of inequality \eqref{gajardo2} for $q=2$, which is easily generalized to any  $q\geq2$,  is in \cite{AST2}.  \\ 
If the graph is compact inequality \eqref{gajardo2} does not hold true (it is clearly violated by constant functions), but  it can be replaced by the weaker inequality 
\begin{equation}\label{gajardo3}
\| \Psi \|_{p} \ \leqslant \ C \| \Psi\|^\alpha_{H^1} \| \Psi \|^{1-\alpha}_q,
\end{equation}
which hold true on any graph  if    $p,q\in[2, +\infty]$, with $p\geq q$, and $\alpha = \frac{2}{2+q}(1-q/p)$,  for all $\Psi\in H^1(\GG)$. \\
A proof of inequality \eqref{gajardo3} for compact graphs is in \cite{Mugnolo}, for non compact graphs it is a trivial consequence of \eqref{gajardo2}. \\
 See also \cite{H} for a collection of useful inequalities on graphs. In what follows we shall always use the weaker inequality \eqref{gajardo3}. \\ 
 
\subsection{Linear Hamiltonian and Quadratic form\label{s:qf}}
We denote by $H$ the Hamiltonian with a $\delta$ coupling of strength $\alpha (v)\in\RE$  at each vertex  and a potential term $W$  on each edge. It is defined as the operator in $L^2(\GG)$ with  domain 
\begin{equation*}
{\mathcal D} (H) :=  \left\{\Psi \in H^2 \textrm{ s.t. }   \sum_{ e\prec \v} \partial_o \psi_e (\v)= \alpha(\v) \psi_e (\v)\quad \forall \underline{ v} \in V
\right\}.
\end{equation*}
where we have denoted by $\partial_o$ the outward derivative from the vertex, it coincides with $\frac{d}{dx}$ or $-\frac{d}{dx}$ according
the orientation on the edge. The action of $H$ is defined by
\begin{equation*}
 (H\Psi)_e = - \psi_e '' + W_e \psi_e,
\end{equation*}
where $W_e$ is the component of the potential $W$ on the edge $e$.\\ 
In the following we will write $V= V_- \cup V_0 \cup V_+ $ where $V_-$, respectively $V_0$, $V_+$, is the set of vertices such that $\al(\v)$ is negative, respectively null, positive.
As recalled in the Introduction,  Assumption \ref{a:1} implies in particular that operator $H$ is a selfadjoint operator on $L^2(\GG)$.
The quadratic form of this operator
is defined on the energy space given by $H^1( \GG)$ and it is explicitly given by
\begin{equation*}
E^{lin}[\Psi ] =  \| \Psi ' \|^2 +(\Psi,W\Psi) +   \sum_{\v\in V} \al(\v) |\Psi(\v)|^2
\end{equation*}
Notice that $\Psi(\v)$ is well defined due to the continuity condition in $H^1(\GG)$. 
\begin{remark}\label{r:train}
Indeed one can prove that under Assumption \ref{a:1} one has 
\begin{equation}\label{train}
\Big|  (\Psi,W\Psi) +    \sum_{\v\in V} \al(\v) |\Psi(\v)|^2\Big| \leq a \|\Psi'\|^2 + b \|\Psi\|^2, \qquad \text{with }0<a<1, \, b>0,
\end{equation}
which, by KLMN theorem, implies that the form $E^{lin}$ is closed and hence defines a selfadjoint operator. It is easy to prove that the corresponding operator coincides with $H$. To prove that the bound \eqref{train} holds true, first note that by Assumption \ref{a:1} we have that $W\in  L^1(\GG) + L^{\infty}(\GG)$. Moreover, by Gagliardo-Nirenberg inequalities, setting $W = W_1 + W_\infty$
\[\begin{aligned}
|(\Psi,W \Psi)|  \leq & \|W_1\|_1 \|\Psi\|_\infty^2  + \|W_\infty\|_\infty \|\Psi\|^2 \\ 
\leq & C \|W_1\|_1\|\Psi\|_{H^1}\|\Psi\|   + \|W_\infty\|_\infty \|\Psi\|^2  \leq  \ve \|\Psi'\|^2 + b_W \|\Psi\|^2,
\end{aligned}\]
where we used the trivial inequality $\|\Psi\|_{H^1}\|\Psi\| \leq  \ve \|\Psi\|^2_{H^1}/2 +  \|\Psi\|^2 /(2\ve) $ for all $\ve > 0$. Similarly, 
\[\left| \sum_{\v\in V} \al(\v) |\Psi(\v)|^2 \right| \leq C_\alpha \|\Psi\|_{\infty}^2\leq C_\alpha \|\Psi\|_{H^1}\|\Psi\|  \leq   \ve \|\Psi'\|^2 + b_\alpha \|\Psi\|^2. 
\]
\end{remark}

Let us define
\[
-E_0 = \inf \left\{ E^{lin} [ \Psi] , \, \Psi\in H^1(\GG),\; \|\Psi\| = 1\right\}.
\]
This corresponds to the bottom of the spectrum of $H$, and it is negative and simple by Assumption \ref{a:2}; we will denote by $\Phi_0$ the corresponding normalized  
 eigenfunction.

\subsection{Linear ground state\label{s:lgs}}
Assumption \ref{a:2} allows to apply bifurcation theory from an eigenvalue in its easiest version and to construct the nonlinear ground state. We stress that there is no obstruction in principle to consider bifurcation from a degenerate eigenvalue but we prefer to avoid unnecessary complications.  However, being not able to indicate a reference where the problem of non degeneracy 
of the ground state on a quantum graph is completely settled, we add some comments on the validity of Assumption \ref{a:2}.\\
Assumption \ref{a:1} with the additional request that the potential $W$ is relatively compact with respect to the laplacian on the graph (Kirchhoff or delta boundary conditions or a mixing of the two) assures that the Hamiltonian $H$ admits an essential spectrum  $\sigma_e(H)=[0,+\infty)$. So that, with this additional condition, a necessary hypothesis for Assumption \ref{a:2} be satisfied is that at least a negative eigenvalue exists.
It is straightforward to prove, considering a trial function constant on the compact part of the graph and smoothly vanishing at infinity that if $ \sum_{\v \in V } |\al (\v) | + \int_{\GG}W$ is negative the quadratic form is negative on this trial function and so a negative eigenvalue exists. Moreover the delta interactions contribute at most with a finite number of eigenvalues and the same holds true if $W_-$ is vanishing sufficiently fast at infinity. The additional request $\int_\GG W(\x) (1+|\x|)d\x<\infty $, as in the line or half line cases is sufficient to guarantee that the discrete spectrum is finite. In particular $-E_0<0$ is an isolated eigenvalue. \\The non degeneracy of the principal eigenvalue
is a subtler problem.  When a ground state exists this property is assured by and is  equivalent to the fact that the heat semigroup $S(t)=\text{exp}(-tH)$ associated to $H$ is positivity improving (see \cite{RSIV}, Thm XIII.44). Moreover, a positivity preserving heat semigroup $S(t)$ is positivity improving, its generator has no ground state degeneracy and its ground state is positive if and only if $S(t)$ is irreducible. The Hamiltonian operator $H_0$, corresponding to the operator $H$ with $W=0$,  generates a positive improving heat semigroup if the quantum graph does not contain tadpoles as subgraphs. This is proven for example in \cite{Mugnolo}, Thm 6.77 for a compact graph and in \cite{KPS} for the general case of non compact graphs. \\ Hence,  when $\inf\sigma(H_0)$ is an eigenvalue, the ground state of a quantum graph without tadpoles and delta boundary conditions at vertices is non degenerate and positive.\\ On the other hand, the absence of tadpoles is not necessary in general, because for example the tadpole graph itself with a delta boundary condition at vertices admits a simple ground state strictly positive, which is explicitly known (see also \cite{cfn15}). 
When $H_0$ is perturbed by the presence of an external potential $W$, it is easy to recognize that the positive part $W_+$ is  harmless and preserve irreducibility of the heat semigroup. If a negative part $W_-$ is present, Thm. XIII.45 in \cite{RSIV} gives a sufficient condition to have irreducibility. A version of this condition suitable for our purposes is given in \cite{Kelleretal} (see in particular Corollary A.3). This result implies that if $W$ is bounded from below and such that $\DD(Q_0+W)$ is dense in $\DD(Q_0)$, where $Q_0$ is the quadratic form of the operator $H_0$, then the heat semigroup generated by $H_0+W$ is positivity improving and  the ground state is non degenerate and positive. \\
We add, by way of information, that simplicity of {\it all} eigenvalues of quantum graph with delta interactions at vertices can be shown to be a generic property up to changing edge lengths and intensity of delta interactions, and again in absence of tadpoles (see \cite{Berko16} for details). 

\subsection{Energy of the nonlinear problem\label{s:enlp}}
The nonlinear energy reads
\[\begin{aligned}
E[\Psi]=& E^{lin}[\Psi ] -\frac{1}{\mu+1} \|\Psi \|_{2\mu+2}^{2\mu + 2}  \\ 
= & \| \Psi ' \|^2 + (\Psi,W\Psi)+  \sum_{\v\in V} \al(\v) |\Psi(\v)|^2  -\frac{1}{\mu+1} \|\Psi \|_{2\mu+2}^{2\mu + 2} 
\end{aligned}
\]
and it is defined on $H^1(\GG)$.
The mass functional is given by
\begin{equation*}
M[\Psi]=\|\Psi\|^2. 
\end{equation*}
Restricted on the mass constraint the nonlinear energy is bounded from below, as a consequence of
the Gagliardo-Nirenberg inequalities on graphs and of the hypotheses on the external potentials, in particular on $W_-$. This is shown in Section 5, at the beginning of proof of Main Theorem.

\subsection{Well-posedness\label{s:wp}}
The local well-posedness for Eq. \eqref{equation} in $H^1({\GG})$ proceeds along well known lines as an application of Banach fixed point theorem. Global well-posedness then follows by conservation laws.\\ 
We will give only a representative result; a more general or optimal result could be obtained by  making use of local in time Strichartz estimates, but we avoid this way  for two reasons. The first one is that our interest in this paper is to establish Th. \ref{t:prob1}, which is a variational property of the NLS on the graph  in $H^1(\GG)$. In the presence of global well-posedness in the same space the existence of ground state implies by well known arguments (see, e.g. \cite{[CL]})  its orbital stability. We do not need deeper or finer results at this level and in any case the picture is clear. \\We stress however that, to the best of our knowledge,  the result given in Th. \ref{t:wp} below  is not present in the literature.\\ The second reason is that Strichartz estimates should be preliminarily proven for starlike  graphs, and this would bring us too far apart. In fact dispersive estimates are known for graphs, but only in some special examples, in particular in trees with Kirchhoff or delta vertices \cite{[BI1], [BI2]} and on the tadpole graph \cite{AMN15}. \\
To proceed we introduce the following integral form of Eq. \eqref{equation}
\begin{equation}
\label{intform1}
\Psi(t) \ = \ e^{-iH t} \Psi_0 + i \int_0^t e^{-i
  H (t-s)}  |\Psi(s)|^{2\mu} \Psi(s)\, ds\equiv {\mathcal T}(\Psi)(t)
\end{equation}
\begin{proposition}[Local well-posedness in $H^1({\GG})$]
\label{loch2}
\n Let $\mu >0$ and  Assumption \ref{a:1} hold true. For any $\Psi_0 \in H^1(\GG)$, there exists $T > 0$ such that the
Eq.  \eqref{intform1} has a unique solution $\Psi \in C ([0,T),
H^1(\GG))
\cap C^1 ([0,T), H^1(\GG)^\star)$. Moreover, Eq. \eqref{intform1} has a maximal solution
defined on an interval of the form $[0, T^\star)$, and the following ``blow-up
alternative''
holds: either $T^\star = \infty$ or
\[
\lim_{t \to T^\star} \| \Psi(t) \|_{H^1(\GG)}
\ = + \infty.
\]
\end{proposition}
\begin{proof}
Consider the space $C([0,T], H^1(\GG)):=C_TH^1$ with the norm $\|\Psi\|_{C_TH^1}=\sup_{[0,T]}\|\Psi\|_{H^1}$ and a closed ball $\overline B_R\subset C_TH^1$. It is well known that $\overline B_{R}$ is a complete metric space. We prove that ${\mathcal T}:\overline B_R \to \overline B_R$ and moreover ${\mathcal T}$ is a contraction on $\overline B_R$ if $R$ and $T$ are suitably chosen.

We start by noting that for any $\Psi \in H^1$ one has that 
\begin{equation}\label{lennon}\|e^{-iHt}\Psi\|_{H^1} \leq C\|\Psi\|_{H^1}.\end{equation}
This inequality  follows from  the conservation of the $L^2$-norm $\|e^{-iHt}\Psi\| =\|\Psi\|$, and from  \[
(1-a)\|(e^{-iHt}\Psi)'\|^2  - b \|e^{-iHt}\Psi\|^2  \leq    E^{lin}[e^{-iHt}\Psi] = E^{lin}[\Psi] \leq (1+a)\|\Psi'\|^2  + b \|\Psi\|^2
\]
where we used the conservation of the linear energy $E^{lin}[e^{-iHt}\Psi] = E^{lin}[\Psi]$, and the bound \eqref{train}.  

By the bound \eqref{lennon},  Schwarz inequality, and  the property of $H^1(\GG)$ of being a Banach algebra
one has
\begin{equation*}
\begin{aligned}
\|{\mathcal T}(\Psi)(t)\|_{H^1}\leq & \left\|\ e^{-iH t} \Psi_0 + i \int_0^t e^{-i
  H(t-s)} |\Psi(s)|^{2\mu} \Psi(s)\ ds \right\|_{H^1}\\
  \leq & C\|\Psi_0\|_{H^1} + C \int_0^t \||\Psi(s)|^{2\mu} \Psi(s)\|_{H^1}\ ds\\
  \leq & C\|\Psi_0\|_{H^1} +  C(\mu)\int_0^t \|\Psi(s)\|^{2\mu+1}_{H^1}\ ds. 
   \end{aligned}
  \end{equation*}
Now, taking the supremum in time
\be
\|{\mathcal T}(\Psi)\|_{C_TH^1}\leq  C\|\Psi_0\|_{H^1} + TC(\mu) \|\Psi\|^{2\mu+1}_{C_TH^1} . 
\ee
We take $R$ such that $C\|\Psi_0\|_{H^1} \leq R/2$,  and in the last inequality we want 
\[TC(\mu) \|\Psi\|^{2\mu+1}_{C_TH^1}
\leq \frac{R}{2}\ .\]
The latter inequality holds true up to taking $T$ small enough, indeed  for $\Psi\in \overline B_R$ one has 
\be
C(\mu) T\|\Psi\|^{2\mu+1}_{C_TH^1}\leq C(\mu) TR^{2\mu+1}\leq \frac{R}{2}
\ee
if  $T\leq\frac{C(\mu) }{2R^{2\mu}}$. 
And this show that ${\mathcal T}:\overline B_R\to \overline B_R$. 

Now we show that we can achieve contractivity of ${\mathcal T}$, possibly choosing a smaller T if needed.\\ We have to bound in $C^TH^1$
\be
{\mathcal T}(\Psi_1)-{\mathcal T}(\Psi_2)=i\int_0^t e^{-iH(t-s)}(|\Psi_1|^{2\mu}\Psi_1-|\Psi_2|^{2\mu}\Psi_2)\ ds.
\ee
By use of mean value theorem one has
\be
\left||\Psi_1|^{2\mu}\Psi_1-|\Psi_2|^{2\mu}\Psi_2\right|\leq C(\mu)(|\Psi_1|^{2\mu}+|\Psi_2|^{2\mu})|\Psi_1-\Psi_2|
\ee
and from this, using again Sobolev immersions in one dimension, 
\be
\||\Psi_1|^{2\mu}\Psi_1-|\Psi_2|^{2\mu}\Psi_2\|_{H^1}\leq C(\mu)(\|\Psi_1\|_{H^1}^{2\mu}+\|\Psi_2\|_{H^1}^{2\mu})\|\Psi_1-\Psi_2\|_{H^1}.
\ee
As before, 
\begin{equation*}
\begin{aligned}
\|{\mathcal T}(\Psi_1)-{\mathcal T}(\Psi_2)\|_{C_TH^1} \leq & \sup_{t\in [0,T]}\left\|\int_0^t e^{-iH(t-s)}(|\Psi_1|^{2\mu}\Psi_1-|\Psi_2|^{2\mu}\Psi_2)\ ds\right\|_{H^1}\\
\leq &\ TC(\mu)(\|\Psi_1\|_{C_TH^1}^{2\mu}+\|\Psi_2\|_{C_TH^1}^{2\mu})\|\Psi_1-\Psi_2\|_{C_TH^1} ,
\end{aligned}
\end{equation*}
and now it is enough to choose   $T$ so small to have $$TC(\mu)(\|\Psi_1\|_{C_TH^1}^{2\mu}+\|\Psi_2\|_{C_TH^1}^{2\mu})<1$$ for $\Psi_1,\Psi_2\in\overline B_R$, which is always possible. \\
The blow-up alternative is shown by bootstrap.\\
For the extension of the solution to $C^1([0,T],H^1(\GG)^\star)\ $ the procedure is similar to the standard case of the equation on $\RE$. Some caution is only needed because of the meaning to give to the equation. One extends first the operator $H$ to $H^1(\GG)$ with values in $H^1(\GG)^\star$ by means of the sesquilinear form $B$ associated to $E$, the (bounded from below) quadratic form of the operator $H$:
\be
( \Psi_1,H\Psi_2 )=B(\Psi_1,\Psi_2)
\ee
as in the standard definition of the weak laplacian.
This allows to show by direct calculation that one has in ${H^1}(\GG)^\star$
\be
\frac{d}{dt}e^{-iH t}\Psi=-iHe^{-iH t}\Psi
\ee
and that a  $C^0 ([0,T),
H^1(\GG))$ solution of Eq. \eqref{intform1} is a $C^1([0,T],H^1(\GG)^\star)$ solution of Eq.  \eqref{equation} and viceversa.

 \end{proof}
\begin{proposition}[Conservation laws]
\n Let $\mu>0$. For any solution $\Psi \in C^0 ([0,T), H^1(\GG))
\cap C^1 ([0,T), H^1(\GG)^\star)$ to
the problem \eqref{intform1}, the following conservation laws hold at
any time $t$:
\begin{equation*}
M[\Psi(t) ] \ = M[ \Psi(0)], \qquad
E[ \Psi(t)] \ = \ E[ \Psi(0) ].
\end{equation*}
\end{proposition}
\noindent
Thanks to the previous theorem and in particular to the fact that $\Psi $ is a $H^1([0,T],H^1(\GG)^\star)$ solution of Eq. \eqref{equation}, the proof is identical to the same proof valid in the standard case of $\RE^n$ and it is omitted.
\begin{theorem}[Global well-posedness] \label{t:wp}
\n Let $0<\mu<2$. For any $\Psi_0 \in H^1(\GG)$,  the
equation \eqref{intform1} has a unique solution $\Psi \in C^0 ([0,\infty),
H^1(\GG) )
\cap C^1 ([0,\infty), H^1(\GG)^\star)$. 
\end{theorem}

\begin{proof}
By Gagliardo-Nirenberg estimates \eqref{gajardo3}, conservation of
the $L^2$-norm and energy, and hypotheses on the potential, one obtains an uniform bound on the $H^1(\GG)$-norm of the solution (see estimate \eqref{e:dec} proven in Section \ref{s:mainth}). So, no blow-up in finite
time can occur, and by the blow-up alternative, the solution is global in time.
\end{proof}

\section{Concentration Compactness lemma}
As noted in \cite{acfn-aihp} where the special case of star graphs was treated, concentration compactness techniques on the real line (or more generally in $\RE^N$) can be adapted to certain domains where translation invariance
is absent. With respect to the classical result (see, e.g., \cite{Caz03, Caz06} for expositions and references) the main point is a finer analysis of the compact case, which  is  split into two sub-cases:
convergent and runaway (see Lem. \ref{l:cc} below). 
In this section we extend the Concentration Compactness lemma to a generic connected noncompact graph with a finite number of internal and external edges. 
In the course of the analysis, where the proofs of single steps require only minor modifications with respect to the standard case, we omit the details and we refer to the already cited texts  \cite{Caz03, Caz06}.\\
We need preliminarily an information about the metric structure of the graph.
We denote by $d(\x,\y)$ the distance between two points of the graph, defined as the infimum of the length of the paths connecting $\x$ to $\y$. 
\begin{proposition}\label{p:de}
Let  $\x=(e,x)\in\GG$, fix the edge $e\in E$ and let $I_e$ be the associated (open) interval, moreover fix a  point  $\y\in\GG$. The  function 
\[d_{e,\y}(x) : I_e \to \RE_+\]
\[d_{e,\y}(x) := d(\x,\y)\]is continuous and piecewise linear. In particular, $d_{e,\y}'$ is a piecewise constant function with at most one discontinuity point $x^*\in I_e$,   and $d_{e,\y}(x)' = 1$ or $d_{e,\y}(x)'=-1$ for all $x\in I_e\backslash\{x^*\}$. 
\end{proposition}
\begin{proof}Assume first that $\y \notin e$. If $e$ is an internal edge (with length $L_e<\infty$),  let $\a$ and $\b$ be the vertices  that identify the endpoints of the edge $e$, note that if $e$ is a loop  $\a$ and $\b$  coincide.  Without loss of generality, set $\a \equiv (e,0)$ and $\b\equiv (e,L_e)$. Then 
\[d_{e,\y}(x) = \min\{d(\a,\y)+x, d(\b,\y)+L_e-x\}.\] 
If $e$ is an external edge, let $\a \equiv (e,0)$ be its endpoint, then   one has 
\[d_{e,\y}(x) =d(\a,\y)+x.\]
On the other hand,  if $\y \in e$, one has  
\[d_{e,\y}(x) =|x-y|.\]
The  properties of $d_{e,\y}$ follow from its explicit form. 
\end{proof}

\n
We denote by $B(\y,t)$ the open ball of radius $t$ and center $\y$
\[
B(\y,t):=\{\x\in\GG \textrm{ s.t. } d(\x,\y)<t\}.
\]
We denote by  $\|\cdot\|_{B(\y,t)}$ the $L^2(\GG)$ norm restricted to the ball
$B(\y,t)$. 

We define the volume of the set $B(\y,t)$ by 
\[ \vol B(\y,t) = \sum_e \int_{I_e} (\one_{B(\y,t)})_e(x) dx \] 
where $\one_{B(\y,t)}$ is the characteristic function of the set $B(\y,t)$.

We have the following bounds on the volume of the sets $B(\y,t)$ and $B(\y,t)\backslash B(\y,s)$:
\begin{proposition}Let $0<s<t<\infty$, then 
\[\vol B(\y,t) \leq 2 N t \qquad \text{and} \qquad \vol \left(B(\y,t)\backslash B(\y,s)\right) \leq 2 N (t-s).\]
\end{proposition}
\begin{proof}We prove only the second bound, the proof of the first one is similar. By definition one has 
\[B(\y,t)\backslash B(\y,s) = \{\x\in\GG \textrm{ s.t. } s\leq d(\x,\y)<t\}, \]
and  
\[ \vol \left(B(\y,t)\backslash B(\y,s) \right)= \sum_e \int_{I_e} (\one_{B(\y,t)\backslash B(\y,s)})_e(x) dx .\] 
 We have that, for each $e\in E$, 
\[ (\one_{B(\y,t)\backslash B(\y,s)})_e(x) =
\left\{\begin{aligned}
&1 \quad && \text{if } s\leq d_{e,\y}(x) <t \\ 
&           0 && \text{otherwise}
\end{aligned}\right.
  \]
By Prop. \ref{p:de},  it is easy to convince oneself that for any edge $e$ 
\[ \int_{I_e} (\one_{B(\y,t)\backslash B(\y,s)})_e(x) dx \leq 2(t-s). \]
From which the  bound on the volume immediately follows. 
\end{proof}

Next we prove a result on the convergence of bounded sequences in $H^1(\GG)$.
\begin{proposition}\label{p:RK}Let $\{\Psi_n\}_{n\in\NA}$ be such that $\Psi_n \in H^1(\GG) $ and $\|\Psi_n\|_{H^1}\leq c$. Then there exists a subsequence $\{\Psi_{n_k}\}_{k\in\NA}$ and a function $\Psi \in H^1(\GG)  $ such that $\Psi_{n_k}\to \Psi$ weakly in $H^1(\GG)$ and 
 $\Psi_{n_k}\to \Psi$ in $L^\infty(B(\underline y,t))$, for any fixed
$\underline y$ and $t$.
\end{proposition}
\begin{proof} 
Since $\Psi_n$ is bounded in $H^1(\GG)$, there exists a subsequence $\Psi_{n_k}$ and a function $\Psi\in H^1(\GG)$, such that $\Psi_{n_k}$ converges to $\Psi$ weakly in $H^1(\GG)$, see, e.g., Th. 2.18 in \cite{LL01}. 

By Gagliardo-Nirenberg inequality the sequence $\Psi_{n_k}$ is uniformly bounded in $L^{\infty}(\GG)$. Then, by  Rellich-Kondrashov theorem, there exists a subsequence, still denoted by $\Psi_{n_k}$, such that $(\Psi_{n_k})_e\to (\Psi)_e$ in $L^\infty(I_e)$   for all the internal edges $e\in E^{in}$, and  $(\Psi_{n_k})_e\to (\Psi)_e$ in $L^\infty(I)$  for all the external edges  $e \in E^{ex}$ and for any bounded subinterval $I$ of $\RE_+$. \\
Moreover, since the functions  $\Psi_n$ are continuous in the vertices, so is $\Psi$ and this concludes the proof of the proposition. 
\end{proof}
\begin{remark}\label{r:RK}As a trivial consequence of Prop. 2.3, one has that the subsequence $\Psi_{n_k}$ convergence to $\Psi$ also in $L^p(B(\underline y,t))$, for all $p\geq1$ and  any fixed
$\underline y$ and $t$.
\end{remark}
For any function $\Psi\in L^2$ and $t\geq0$ we define the concentration function
$\rho(\Psi,t)$ as
\beq
\label{e:rho}
\rho(\Psi,t) = \sup_{\y\in\GG} \|\Psi\|_{B(\y,t)}^2\,.
\eeq

In the following proposition we prove two important properties of the
concentration function: that the $\sup$ at the r.h.s. of equation \eqref{e:rho}
is indeed attained at some point of $\GG$ and the H\"older continuity of
$\rho(\Psi,\cdot)$.

\begin{proposition}
\label{p:rho}
Let $\Psi\in L^2$ be such that $\|\Psi\|>0$, then
\begin{enumerate}[i)]
\item{\label{i:rho0}}
$\rho(\Psi,\cdot)$ is
non-decreasing,  $\rho(\Psi,0)=0$, $0<\rho(\Psi,t)\leq M[\Psi]$ for $t>0$,
and $\lim_{t\to\infty} \rho(\Psi,t) =  M[\Psi]$. 
\item{\label{i:rho1}}
 There exists $\y(\Psi,t)\in\GG$ such that 
\[
\rho(\Psi,t)= \|\Psi\|_{B(\y(\Psi,t),t)}^2\,.
\]
\item
\label{i:rho2}
If $\Psi\in L^p$ for some $2\leq p\leq\infty$, then 
\beq
\label{e:ii}
|\rho(\Psi,t)-\rho(\Psi,s)| \leq c \|\Psi\|_{p}^2 |t-s|^{\frac{p-2}{p}} \qquad
\textrm{for } 2\leq p<\infty,
\eeq
and 
\beq
\label{e:ii-infty}
|\rho(\Psi,t)-\rho(\Psi,s)| \leq c \|\Psi\|_{\infty}^2 |t-s| ,
\eeq
for all $s,t>0$ and where $c$ is independent of $\Psi$, $s$ and $t$.
\end{enumerate}
\end{proposition}
\begin{proof}
The proofs of \emph{\ref{i:rho0})} and \emph{\ref{i:rho1})} follow directly from the proof of Lem. 1.7.4 in \cite{Caz03}. 

To prove  \emph{\ref{i:rho2})} one uses the inequality 
\begin{align*}
|\rho(\Psi,t)-\rho(\Psi,s)| \leq \|\Psi\|_{B(\y(\Psi,t),t)\backslash B(\y(\Psi,t),s)}^2,
\end{align*}
see Lem. 1.7.4 in \cite{Caz03},  and the inequalities: 
\[\|\Psi\|_{B(\y,t)\backslash B(\y,s)}^2 \leq \|\Psi\|^2,
\]
for $p=2$; 
\[\|\Psi\|_{B(\y,t)\backslash B(\y,s)}^2 \leq
[\vol (B(\y,t)\backslash B(\y,s))]^{\frac{p-2}{p}}  \|\Psi\|_{p}^2 \leq 
(2N|t-s|)^{\frac{p-2}{p}} \|\Psi\|_{p}^2,
\]
for $2<p<\infty$; 
and 
\[\|\Psi\|_{B(\y,t)\backslash B(\y,s)}^2 \leq 2N |t-s| \|\Psi\|_{\infty}\]
for $p=\infty$. 
\end{proof}

For any sequence $\Psi_n \in L^2$ we define the concentrated mass
parameter $\tau$ as  
\begin{equation*}
\tau = \lim_{t\to\infty} \liminf_{n\to\infty} \rho(\Psi_n,t)\,.
\end{equation*}
Te parameter  $\tau$ plays a key role in the
concentration compactness lemma because it distinguishes the
occurrence of vanishing, dichotomy or compactness in $H^1(\GG)$-bounded
sequences. The following lemma (see for the standard case
Lem. 1.7.5 in  \cite{Caz03}), proves that $\tau$ can be computed as
the limit of $\rho$ on a suitable subsequence.

\begin{lemma}
\label{l:175}
Let $m>0$ and  $\{\Psi_n\}_{n\in\NA}$ be such that: $\Psi_n\in H^1(\GG)$,
\beq
\label{e:end-1}
M[\Psi_n] \to m\quad \text{as} \quad n\to \infty\,,
\eeq
and
\beq
\label{e:end-2}
\sup_{n\in\NA}\|\Psi_n'\|<\infty\,.
\eeq
Then there exist a subsequence $\{\Psi_{n_k}\}_{k\in\NA}$, a nondecreasing
function $\gamma(t)$, and a sequence $t_k\to\infty$ with the following
properties:
\begin{enumerate}[i)]
\item
\label{i:grass-1}
 $\rho(\Psi_{n_k},\cdot)\to \gamma(\cdot)\in [0,m]$ as $k\to\infty$
  uniformly on bounded sets of $[0,\infty)$.
\item 
\label{i:grass-2}
$\tau =
\lim_{t\to\infty}\gamma(t)=\lim_{k\to\infty}\rho(\Psi_{n_k},t_k)=\lim_{
k\to\infty}\rho(\Psi_{n_k},t_k/2)$.
\end{enumerate}
\end{lemma}
\begin{proof} We refer to \cite[Lem. 1.7.5]{Caz03} for the details of the proof. Here we just remark that the equicontinuity of the sequence $\rho(\Psi_{n_k},\cdot)$, needed to apply Arzel\`a - Ascoli theorem, follows from  \eqref{e:ii-infty}, and from the fact that, by Gagliardo-Nirenberg inequality and assumptions
\eqref{e:end-1} - \eqref{e:end-2}, $\|\Psi_n\|_{\infty}$ is uniformly bounded in $n$. 
\end{proof}

We are now ready to prove the concentration compactness lemma.  Although the  statement of the lemma is similar both to the standard case (see \cite[Prop.1.7.6]{Caz03}) and to Lem. 3.3 in \cite{acfn-aihp} where the case of star graph is treated, its proof requires several adjustments and changes and for this reason we provide all the details. We also remark that  the argument  used here to prove the existence of \emph{runaway sequences}  is  simpler than the one used in \cite{acfn-aihp}. 
\begin{lemma}[Concentration compactness]
\label{l:cc}
Let $m>0$ and $\{\Psi_n\}_{n\in\NA}$ be such
that: $\Psi_n\in H^1(\GG)$,
\begin{equation}\label{supfun}
M[\Psi_n] \to m \quad \text{as} \quad n\to \infty\,,
\end{equation}
\begin{equation}\label{supder}
\sup_{n\in\NA}\|\Psi_n'\|<\infty\,.
\end{equation}
Then there exists a subsequence $\{\Psi_{n_k}\}_{k\in \NA}$ such that:
\begin{enumerate}[i)]
\item\label{i:cc1}
(Compactness) If $\tau=m$, at least one of
    the two following cases occurs:
\begin{itemize}
\item[$i_1)$](Convergence) There exists a
function
$\Psi\in H^1(\GG)$  such that $\Psi_{n_k}\to \Psi$  in $L^p$ as
$k\to\infty$ for all $2\leq p\leq \infty$ .
\item[$i_2)$]\label{i:runaway}(Runaway) There exists $e^*\in E^{ex}$, such that for any $t>0$,  and
$2\leq p\leq\infty$
\beq
\label{e:run-1}
\lim_{k\to \infty}  
\left(\sum_{e\neq e^*}\|(\Psi_{n_k})_e\|_{L^p(I_e)}^p   + \|(\Psi_{n_k})_{e^*}\|_{L^p((0,t))}^p \right) =0 .
\eeq
\end{itemize}
\item\label{i:cc2} (Vanishing) If $\tau=0$, then  $\Psi_{n_k}\to 0 $ in $L^p$ as
$k\to\infty$  for all $2< p\leq \infty$. 
\item\label{i:cc3} (Dichotomy) If $0<\tau<m$, then there exist two sequences 
$\{\VV_k\}_{k\in\NA}$ and $\{\WW_k\}_{k\in\NA}$ in $H^1(\GG)$
such that 
\beq
 \supp \VV_k \cap \supp \WW_k = \emptyset 
\label{dic1} 
\eeq
\beq
|\VV_k(\x)| + |\WW_k(\x)| \leq |\Psi_{n_k}(\x)| \qquad \forall \x\in\GG 
\label{dic2} 
\eeq
\beq
\| \VV_k \|_{H^1(\GG)} + \|\WW_k\|_{H^1(\GG)} \leq c
\|\Psi_{n_k}\|_{H^1(\GG)}
\label{dic3} 
\eeq
\beq
\lim_{k \to \infty} M [\VV_k ] = \tau \qquad \qquad \lim_{k \to \infty} M[
\WW_k ]= m -\tau
\label{dic4} 
\eeq
\beq
\liminf_{k\to \infty} \lf( \|\Psi_{n_k}'\|^2 - \| \VV_k' \|^2 - \|
\WW_k' \|^2 \ri) \geq 0
\label{dic5}
\eeq
\beq
\lim_{k\to \infty} \lf( \|\Psi_{n_k} \|_{p}^p - \| \VV_k \|_{p}^p - \|
\WW_k \|_{p}^p \ri) =0 \qquad 2 \leq p < \infty
\label{dic6}
\eeq
\beq
\label{dic7}
\lim_{k\to\infty}\left\||\Psi_{n_k}|^2-  |\VV_{k}|^2 -| \WW_{k}|^2\right\|_\infty=0.
\eeq
\end{enumerate}
\end{lemma}
\begin{proof}
Let $\{\Psi_{n_k}\}_{k\in\NA}$, $\gamma(\cdot)$ and $t_k$  be the subsequence, the
function and the sequence defined in Lem. \ref{l:175}.  

Proof of \emph{\ref{i:cc1})}. Suppose $\tau=m$. By Lem. \ref{l:175}
\emph{\ref{i:grass-2})},  for any  $m/2\leq\lambda< m$ there exists $t_\la$
large enough such that $\gamma(t_\la)>\la$. Then by Lem.  \ref{l:175}
\emph{\ref{i:grass-1})}, for $k$ large enough $\rho(\Psi_{n_k},t_\la)>\la$. 

Set $\underline{y_k}(t)\equiv \underline{y}(\Psi_{n_k},t)$, where
$\underline{y}(\Psi_{n_k},t)$ was defined in Prop. \ref{p:rho}
\emph{\ref{i:rho1})}. For $k$ large enough, we have that 
\beq
\label{e:cmin}
d(\underline{y_{k}}(t_{m/2}),\underline{y_{k}}(t_{\la})) \leq t_{m/2}+t_\la\,. 
\eeq
To prove \eqref{e:cmin}, assume that 
$d(\underline{y_{k}}(t_{m/2}),\underline{y_{k}}(t_{\la})) > t_{m/2}+t_\la$,
then the balls $B(\underline{y_k}(t_{m/2}),t_{m/2})$ and
$B(\underline{y_k}(t_\la),t_{\la})$ would be disjoint, thus implying  
\[
M[\Psi_{n_k}]\geq \|\Psi_{n_k}\|_{B(\underline{y_k}(t_{m/2}),t_{m/2})}^2 +
\|\Psi_{n_k}\|_{B(\underline{y_k}(t_\la),t_{\la})}^2 > \frac{m}{2}+\la \geq m
\]
which is impossible because $M[\Psi_{n_k}]\to m$.
Next we distinguish two cases: $\{\underline{y_{k}}(t_{m/2})\}_{k\in\NA}$ bounded (it belongs to a finite ball on the graph)
and $\{\underline{y_k}(t_{m/2})\}_{k\in\NA}$ unbounded  (there is no finite  ball on the graph containing the sequence). 

Case $\underline{y_{k}}(t_{m/2})$ bounded.  By Prop. \ref{p:RK} and Rem. \ref{r:RK}, we have that there exists a subsequence $\Psi_{n_k}$ and a function $\Psi\in H^1(\GG)$ such that  $\Psi_{n_k}\to \Psi$ weakly in $H^1(\GG)$ and $\Psi_{n_k}\to \Psi$ in $L^2(B(\underline y,t))$ for any fixed
$\underline y$ and $t$.  \\ 
The function $\Psi$ might be the null function, next we show that for $\underline{y_{k}}$ bounded this is not the case. We prove indeed that $M[\Psi]=m$ which, together with the weak convergence in $H^1(\GG)$, implies that  $\Psi_{n_k} \to
\Psi$ in $L^2$, then the convergence in $L^p$ for $2<p\leq \infty$ follows from
Gagliardo-Nirenberg inequality. 
\\
Fix $\lambda \in (m/2,m)$, and let
  $t_\lambda$ be such that $\rho (\Psi_{n_k}, t_\lambda) > \lambda$ for $k$ large enough. Since, by \eqref{e:cmin},
  $\underline{y_{k}} (t_\lambda)$ is bounded as well,  up to choosing a
subsequence which we  still denote by $\Psi_{n_k}$, we
can
assume that $\underline{y_{k}}(t_\lambda)\to
\underline{y}^*(t_\lambda)$ and  $\underline{y_{k}}(t_{m/2})\to
\underline{y}^*(t_{m/2})$.  Then, for any fixed $\ve>0$ and $k$ large enough
we
have $d(\underline{y}^*(t_{m/2}),\underline{y_k}(t_{m/2}))\leq
\ve$, so that, by
\eqref{e:cmin} and the triangle inequality it follows that  $d(\underline{y}^*(t_{m/2}),
\underline{y_k}(t_\la)) \leq \ve + t_{m/2} + t_\la$. Setting $T= 2(\ve +
t_{m/2} + t_\la) $ we certainly have that
$B(\underline{y_k}(t_\la),t_\la)\subseteq B(\underline{y}^*(t_{m/2}),T)$ so that
\begin{equation*}
\|\Psi_{n_k}\|_{B(\underline{y}^*(t_{m/2}),T)}^2 \geq  
\|\Psi_{n_k}\|_{B(\underline{y_k}(t_{\la}),t_\la)}^2 =
\rho(\Psi_{n_k},t_\la)>\la\,. 
\end{equation*}
Since 
\[
M[\Psi] \geq \|\Psi\|_{B(\underline{y}^*(t_{m/2}),T)}^2 =\lim_{k\to\infty}
\|\Psi_{n_k}\|_{B(\underline{y}^*(t_{m/2}),T)}^2,
\]
 we have that  $M[\Psi] \geq \la$. As we can
   choose $\lambda$ arbitrarily close to $m$, we get $M[\Psi] \geq
m$. On the other hand, by weak convergence,  we have that 
\[
M[\Psi] \leq \liminf_{k\to\infty}M[\Psi_{n_k}]=m,
\]
so that $M[\Psi] =m$. 
 
Assume now that  $\underline{y_{k}}(t_{m/2})$ is
unbounded. Then,  up to choosing a subsequence, which we still
denote by $\Psi_{n_k}$, we
can
assume that there exists $e^*\in E^{ex}$ such that  $ \{\underline{y_{k}}(t_{m/2})\}_{k\in\NA}$ belongs to the the edge $ e^*$ and
$y_k(t_{m/2})\to\infty$.\\
Fix $\ve$ and $t$. Set $\lambda = m-\ve$ and  $t_\la$ such that for $k$ large enough 
$\rho(\Psi_{n_k},t_\la)>\la$. 
 By \eqref{e:cmin} we have that $y_k(t_\la)\to \infty $, so that,  for 
$k$ large enough, $y_k(t_\la)-t_\la>t$ and  
\[\begin{aligned}
 \int_{t}^\infty |(\Psi_{n_k})_{e^*} (x)|^2 dx \geq &  \|(\Psi_{n_k})_{e^*} \|^2_{L^2((y_k(t_\la)-t_\la,y_k(t_\la)+t_\la))}   \\ 
 = & \|\Psi_{n_k}\|_{B(\underline{y_k}(t_\la),t_\la)}^2 =
 \rho(\Psi_{n_k},t_\la) >\la =m-\ve.
\end{aligned}\]
On the other hand, by \eqref{supfun} and for $k$ large enough, one has that  
\[
M[\Psi_{n_k}] = \sum_{e\neq e^*} \|(\Psi_{n_k})_{e}  \|^2_{L^2(I_e)} +
 \int_{0}^t |(\Psi_{n_k})_{e^*} (x)|^2 dx + \int_{t}^\infty |(\Psi_{n_k})_{e^*} (x)|^2
dx
< m+\ve,
\]
so that 
\[
\sum_{e\neq e^*} \|(\Psi_{n_k})_{e^*} \|^2_{L^2(I_e)} +
 \int_{0}^t |(\Psi_{n_k})_{e^*} (x)|^2 dx
< 2\ve \,.
\]
The limit \eqref{e:run-1} for $p>2$
follows by   Gagliardo-Nirenberg inequalities applied to the graph $\GG_t$ obtained from $\GG$ by cutting the edge $e^*$ at length $t$. We remark that the graph $\GG_t$ might be  compact. 
\\

Proof of \emph{\ref{i:cc2})}. We start with the proof of a useful inequality, see Eq. \eqref{a} below. Let $L_{max}$ be the maximal length of the internal edges. For any internal edge $e\in E^{in}$,   by Gagliardo-Nirenberg inequality applied to the interval $I_e$, by Eq. \eqref{e:rho},  and  since $\rho(\Psi,\cdot)$
is non-decreasing, one has that 
\[\begin{aligned}
\|\psi_e\|_{L^6(I_e)}^6 \leq & c_e \|\psi_e\|_{L^2(I_e)}^4 \| \psi_e\|_{H^1(I_e)}^2  \\ 
  \leq & c_e \rho(\Psi,L_e/2)^2  \| \psi_e\|_{H^1(I_e)}^2
 \leq c \rho(\Psi,L_{max}/2)^2  \| \psi_e\|_{H^1(I_e)}^2
\end{aligned}\] 
where $c_e$ is a constant that depends on the edge $e\in E^{in}$ (on the length of the interval $I_e$)  and we set $c = \max_{e\in E^{in}} c_e$.  On the other hand, for any external edge $e\in E^{ext}$, one has 
\[\begin{aligned}
\|\psi_e\|_{L^6(\RE_+)}^6 
= &  \sum_{n = 0}^{\infty} \|\psi_e \|^6_{L^6((nL_{max},(n+1)L_{max}))} \\ 
\leq &c \sum_{n = 0}^{\infty} \|\psi_e \|^4_{L^2((nL_{max},(n+1)L_{max}))}\|\psi_e \|^2_{H^1((nL_{max},(n+1)L_{max}))} \\
\leq &c\rho(\Psi,L_{max}/2)^2  \sum_{n = 0}^{\infty}\|\psi_e \|^2_{H^1((nL_{max},(n+1)L_{max}))} = c\rho(\Psi,L_{max}/2)^2 \|\psi_e \|^2_{H^1(\RE_+)} ,
\end{aligned}\]
where $c$ is a constant that depends on $L_{max}$. Summing up on internal and external edges we get 
\begin{equation}
\label{a}
\|\Psi\|_{6}^6  \leq  c\rho(\Psi,L_{max}/2)^2 \|\Psi \|^2_{H^1}.  
\end{equation}
Suppose now that  $\tau=0$.  By Lem. \ref{l:175},
$\tau=\lim_{k\to\infty}\rho(\Psi_{n_k},t_k)=0$.  Then since $\rho(\Psi,\cdot)$
is non-decreasing and $t_k\to \infty$, $\lim_{k\to\infty}\rho(\Psi_{n_k},L_{max}/2)=0$, and $\lim_{k\to\infty}\|\Psi_{n_k}\|_{6}=0$ by \eqref{a}.
The statement   for $2<p<6$ follows from the  H\"older
inequality $\|\Psi\|_{p} \leq
\|\Psi\|_{6}^{\frac{3(p-2)}{2p}}\|\Psi\|^{\frac{(6-p)}{2p}}$, while  for
$6< p\leq  \infty$ one uses  inequality   \eqref{gajardo3} with $q=6$.\\

Proof of \emph{\ref{i:cc3})}. Suppose that $0<\tau<m$ and let  $\theta$ and $\varphi$ be two cut-off functions such that $\theta,\varphi\in C^{\infty} (\RE_+ )$, 
$0\leqslant \theta , \varphi \leqslant 1$ and 
\[
\theta(t) = 
\begin{cases}
1 & 0\leqslant t \leqslant 1/2 \\
0 & t \geqslant 3/4 
\end{cases}
\qquad \qquad 
\varphi(t) = 
\begin{cases}
0 & 0\leqslant t \leqslant 3/4 \\
1 & t \geqslant 1 
\end{cases} 
\]
Set $\underline{y}(t_k)\equiv \underline{y}(\Psi_{n_k},t_k)$, where
$\underline{y}(\Psi_{n_k},t)$ was defined in Prop. \ref{p:rho}
\emph{\ref{i:rho1})}. Define  the following cut off functions
\begin{equation}\label{thetaphi}
\Theta_k (\x) = \theta \lf( \frac{d(\x ,  \y ( t_k /2)) }{t_k} \ri) \qquad \qquad 
\Phi_k (\x) = \varphi \lf( \frac{d (\x , \y ( t_k /2) ) }{t_k} \ri) .
\end{equation}
We remark that $(\Theta_k)_e (x) = \theta ( d_{e, \y ( t_k /2)}(x)/t_k)$ with $d_{e,\y}$ given as in Prop. \ref{p:de}, and similarly for $\Phi_k$. \\ 
Let $\VV_k$ be defined by 
\[
\VV_k (\x ) = \Theta_k ( \x) \Psi_{n_k}(\x),
\]
and let  $\WW_k$  be  defined by
\[
\WW_k (\x ) = \Phi_k ( \x) \Psi_{n_k}(\x),
\]
products to be understood pointwise. We remark that $\VV_k$ ($\WW_k$ resp.) coincides with $\Psi_{n_k}$ in
the ball $B(\underline y(t_k/2),t_k/2)$ (in the set $\GG\backslash
B(\underline y(t_k/2),t_k)$ resp.) and $\VV_k = 0$ ($\WW_k = 0$ resp.)
in  the set $\GG\backslash B(\underline y(t_k/2),3t_k/4)$ (in the ball
$B(\underline y(t_k/2),3t_k/4)$ resp.). Properties \eqref{dic1} and
\eqref{dic2}  are immediate. Property \eqref{dic3} also immediately follows from the definitions of $R_k$ and $S_k$ and from Prop. \ref{p:de}.  Next we notice that by
Prop. \ref{p:rho}, \emph{\ref{i:rho1})}, 
\[
\rho ( \Psi_{n_k} , t_k /2 )  = \|\Psi_{n_k}\|_{B(\underline y(t_k/2),t_k/2)}^2 
\leqslant M[\VV_k]. 
\]
Moreover, since  $\theta(t) \leq 1$,  
\[
M[\VV_k] \leq  \|\Psi_{n_k}\|_{B(\underline y(t_k/2),t_k)}^2 \leq \|\Psi_{n_k}\|_{B(\underline y(t_k),t_k)}^2 = \rho( \Psi_{n_k} , t_k)\,,
\]
where we have taken into account 
 the optimality of $y(t_k)$ according to Prop. \ref{p:rho}, \emph{\ref{i:rho1})} and the definition of $\rho(\Psi,t)$. Therefore 
\be
\lim_{k\to \infty}M[\VV_k] = \tau
\ee
by Lem. \ref{l:175}, \emph{\ref{i:grass-2})}. Define  $\ZZ_k := \Psi_{n_k} - \VV_k - \WW_k $ and notice that 
\[\supp (\ZZ_k) \subseteq B(\underline y(t_k/2),t_k)\backslash B(\underline y(t_k/2),t_k/2)\]
 and  $|\ZZ_k|\leqslant |\Psi_{n_k}|$, to be understood pointwise.
Then one has 
\begin{align}
M[\ZZ_k] \leq &  \|\Psi_{n_k}\|^2_{B(\underline y(t_k/2),t_k)\backslash B(\underline y(t_k/2),t_k/2)}
\nonumber
\\
= & \|\Psi_{n_k}\|^2_{B(\underline y(t_k/2),t_k)} - 
     \|\Psi_{n_k}\|^2_{B(\underline y(t_k/2),t_k/2)}
\leq \rho(\Psi_{n_k} , t_k ) - \rho(\Psi_{n_k} , t_k /2) 
\label{zampone}
\end{align}
again by the optimality properties of $\underline y(t_k)$. It follows from \eqref{zampone} and Lem. \ref{l:175}, \emph{\ref{i:grass-2})}  that 
\begin{equation}\label{ZZk}
M[\ZZ_k]\to 0  \quad \text{as }k \to \infty, 
\end{equation} 
and therefore $M[\WW_k]\to m -\tau$ which concludes the proof of \eqref{dic4}.

To prove \eqref{dic6} and \eqref{dic7} we use 
\begin{equation}\label{radio}
\big|
| \Psi_{n_k}(\x)|^p - |\VV_k(\x)|^p -|\WW_k(\x)|^p
\big| \leqslant c_p |\Psi_{n_k}(\x)|^{p-1} |\ZZ_k(\x)| \qquad p\geq 1,
\end{equation}
to be understood pointwise, which in turn implies 
\[
\big|\|\Psi_{n_k}\|^p - \|\VV_k\|^p -\|\WW_k\|^p
\big| \leqslant c \|\Psi_{n_k}\|_{2(p-1)}^{p-1} \|\ZZ_k\| \leq c \|\ZZ_k\| \qquad p\geq 2
\]
where we used \eqref{supfun}, \eqref{supder}, and Gagliardo-Nirenberg inequality \eqref{gajardo3}. The limit  \eqref{dic6} then follows from $\|\ZZ_k\|\to 0$. To prove \eqref{dic7} we use \eqref{radio} with $p=1$, and the fact that, by $\|\ZZ_k\|_{H^1}\leq c$, $\|\ZZ_k\|\to 0$,  and Gagliardo-Nirenberg inequality, one has $\|\ZZ_k\|_\infty \to 0$. 

Concerning the  inequality \eqref{dic5},
first notice that
\[
\begin{aligned}
&|(\Psi_{n_k})_e '|^2 - |(\VV_k)_e ' |^2 - |(\WW_k)_e '|^2 \\
=  &|(\Psi_{n_k})_e ' |^2 \big[1-(\Theta_k)_e^2 -(\Phi_k)_e^2\big]  \\ 
&-|(\Psi_{n_k})_e|^2 \big[((\Theta_k)_e ')^2 + ((\Phi_k ')_e)^2\big]
- \Re\overline{(\Psi_{n_k})}_e(\Psi_{n_k})_e'\big[(\Theta_k)_e^2 +(\Phi_k)_e^2 \big]' \\
\geqslant &-\frac{c}{t_k^2} |(\Psi_{n_k})_e|^2
- \frac{c}{t_k} |(\Psi_{n_k})_e' | |(\Psi_{n_k})_e|
\end{aligned}
\]
for almost all $x\in I_e$, where we used $1-(\Theta_k)_e^2 -(\Phi_k)_e^2\geq 0$ and the fact that  $|(\Theta_k)_e '(x)|\leq c/t_k$, $|(\Phi_k ')_e(x)|\leq c/t_k$ for almost all $x\in I_e$ (see the remark below Eq. \eqref{thetaphi} and Prop. \ref{p:de}).  The inequality \eqref{dic5} follows by  integrating on $I_e$ and summing up on $e$, and by recalling that $t_k\to\infty$. 
\end{proof}
\begin{remark}\label{r:2.8}
We note that Eq. \eqref{e:run-1} in Lem. \ref{l:cc}-$i_2)$ implies that in the runaway case 
\[\lim_{k\to \infty}  
\|\Psi_{n_k}\|_{L^p(B(\y,t))} =0 \]for any $2\leq p\leq \infty $,  $\y\in\GG$,  and $t>0$. 
\end{remark}

\section{Bifurcation analysis}
\n
In this section we study the solutions of 
\beq \label{eqstaz}
H \Phi - |\Phi|^{2\mu} \Phi = - \ome \Phi  \qquad \Phi\in \DD(H),\;\ome>0,
\eeq
by means of bifurcation techniques. Eq. \eqref{eqstaz} is the stationary equation associated to Eq. \eqref{equation}, it arises when one looks for solutions of the form $\Psi(t) = e^{i\omega t}\Phi(\omega)$.

The main result of this section is the following:
\begin{theorem}[Bifurcation from the linear ground state]\label{t:bif}
 If Assumptions \ref{a:0}, \ref{a:1}, and \ref{a:2} hold true, then there exists $\delta>0 $ such that, for any $\omega\in(E_0, E_0+\delta )$, Eq. \eqref{eqstaz} admits a unique (up to phase multiplication) solution $\Phi(\omega)$. Moreover, the function $m(\omega):= \|\Phi(\omega)\|^2 $ belongs to  $C^1(E_0,E_0+\de)$, is such that 
\begin{equation}\label{mPhiome}
m(\omega) =  \left( \frac{ \ome-E_0 }{ \| \Phi_0\|_{2\mu+2 }^{2\mu+2 } } \right)^{\f{1}{\mu} } + o\left((\omega-E_0)^{\frac1\mu}\right),
\end{equation}
and it is  invertible. Denoting its inverse by $\omega(m)$, one has that the function  $E(m) := E[\Phi(\ome(m))]$ is continuous for $m>0$ small enough,  and  
\begin{equation}\label{EPhiome}
E(m)  =-E_0 m + o(m). 
\end{equation}
\end{theorem}
\begin{proof}
We follow the approach used in \cite{kirr-kevrekidis-pelinovsky:11}. Without loss of generality, we can take $\Phi(\omega)$  real valued. We start by noting  that $\DD(H)$ with the graph norm  $\vertiii{\Phi}_H =  \|H \Phi \|+\|\Phi\|$    is a Banach space. 

We note that the following inequality holds true: 
\begin{equation*}
\||\Psi|^{2\mu} \Psi \| \leq C \vertiii{\Psi}_H^{2\mu+1}.
\end{equation*}
To prove it we use first H\"older and Gagliardo-Nirenberg inequalities  to obtain  \[\||\Psi|^{2\mu} \Psi \| \leq C \|\Psi'\|^\mu \vertiii{\Psi}^{\mu+1}_H.\] Then we prove that $\|\Psi'\| \leq C  \vertiii{\Psi}_H$. To this aim,   we use the fact that $E[\Psi] = (\Psi,H\Psi)$, which in turn implies 
\begin{align}
 \|\Psi'\|^2 \leq  & (\Psi,H\Psi) + (\Psi,W_-\Psi) + \sum_{\v \in V_- } |\al(v)| |\Psi(\v)|^2  \nonumber  \\ 
 \leq  &  \vertiii{\Psi}_H^2 + C_0\left( \|\Psi'\|^\frac1r \vertiii{\Psi}_H^{2-\frac1r}  + \|\Psi'\|\vertiii{\Psi}_H\right), \label{singer}
\end{align}
 for some large constant $C_0$. Here we used again Gagliardo-Nirenberg inequality and 
\begin{equation*}
(\Psi,W_- \Psi)  \leq \|W_-\|_r \|\Psi\|^2_{2r/(r-1)} \leq C \|\Psi'\|^\frac1r \vertiii{\Psi}_H^{2-\frac1r} , \end{equation*}
see also Eq. \eqref{W-1} below.  By the bound \eqref{singer} we infer that if $\|\Psi'\| > 4C_0  \vertiii{\Psi}_H$ it must be $ \|\Psi'\| \leq   \vertiii{\Psi}_H$, hence $\|\Psi'\| \leq C  \vertiii{\Psi}_H$.

Let us introduce the map $F: \DD (H)\times \RE_+ \rightarrow L^2$
\begin{equation}\label{arms}
F(\Phi , \ome) = (H+\ome ) \Phi -| \Phi|^{2\mu} \Phi.
\end{equation}
Note that despite the fact that we keep using the notation $| \Phi|^{2\mu} \Phi$, since we are assuming $\Phi$ real valued $| \Phi|^{2\mu} \Phi$ could be understood as $(\Phi)^{2\mu+1}$.   It is clear that $F\in C^1( \DD (H)\times \RE_+ , L^2)$. Notice that
\be 
D_\Phi F (0, \ome) \Psi =  (H+\ome ) \Psi.
\ee
We use the Lyapunov-Schmidt method to study the existence of solutions of
\beq \label{bifgen}
F(\Phi,\ome) =0,
\eeq
which is equivalent to Eq. \eqref{eqstaz}. Let us introduce two orthogonal projectors in $L^2$
\be
P= \Phi_0 \left(\Phi_0 ,\,\cdot\, \right)  \qquad \qquad Q= I-P
\ee
and extend their action on $\DD(H)$.
 We decompose accordingly
\be
\Phi= a \Phi_0 + \Theta, 
\ee
where $a = \left(\Phi_0,\Phi\right)$ and $Q\Theta =0$. This decomposition is well defined on $\DD(H)$, i.e., $\Theta\in D(H)$. Moreover, since if $\Phi$ is a solution of Eq. \eqref{arms} so is $-\Phi$ we can assume $a\geq0$. Then equation \eqref{bifgen} is equivalent to the system
\beq \label{LS}\left\{
\begin{aligned}
& QF( a \Phi_0 + \Theta ,\ome) =0 \\
& PF( a \Phi_0 + \Theta ,\ome) =0 
\end{aligned}
\right.
\eeq
The first equation in \eqref{LS} is called \emph{auxiliary equation}, the second one is the \emph{bifurcation equation}. We introduce  the map $G:\RE_+\times \DD(H)\times \RE_+\to L^2$
\begin{equation*}
G(a, \Theta , \ome) := QF( a \Phi_0 + \Theta ,\ome) = Q  (H+\ome ) Q \Theta - Q|  a \Phi_0 + \Theta |^{2\mu } (  a \Phi_0 + \Theta),
\end{equation*}
hence, the auxiliary equation is equivalently written as $G(a, \Theta , \ome) = 0$.  Since
\be
D_\Theta G (0,0, E_0 ) = Q (H+ E_0) Q
\ee
is invertible then, by the Implicit Function Theorem in Banach spaces, the auxiliary equation defines locally in a neighborhood $I = (0,\ve)  \times (E_0-\delta,E_0+\delta) $ a unique function  $\Theta_* (a, \ome)$ in $C^1(I,\DD(H))$ such  that $Q\Theta_* = \Theta_*$, 
\be 
QF( a \Phi_0 +\Theta_* (a, \ome) ,\ome) =0,
\ee
and $\lim_{(a,\omega)\to(0^+,E_0)} \vertiii{\Theta_*(a,\omega)}_H = 0$. Indeed from the equation $G(a,\Theta_*(0,\omega),\omega) = 0$, and since $Q  (H+\ome ) Q $  is invertible with bounded inverse  in a neighborhood of $E_0$, one has that 
\[
\begin{aligned}
\vertiii{\Theta_*(a,\omega)}_H \leq & C \left( \||a\Phi_0|^{2\mu} a\Phi_0\| + \||\Theta_*(a,\omega)|^{2\mu}\Theta_*(a,\omega)\| \right) \\
\leq &  C \left( \vertiii{a\Phi_0}^{2\mu+1}_H + \vertiii{\Theta_*(a,\omega)}^{2\mu+1}_H \right) .
\end{aligned}
\] 
Since taking $\ve$ and $\delta$ small enough  we can make  $\vertiii{\Theta_*}_H$ arbitrarily small we  have
\begin{equation} \label{thetabdd} 
\vertiii{\Theta_* (a, \ome) }_{H} \leq C a^{2\mu+1}.
\end{equation}

\n
Now we turn our attention to the bifurcation equation. First we write it explicitly using the definition of $P$.
\begin{equation*}
a (\ome -E_0 ) - \Big( \Phi_0 , | a \Phi_0 + \Theta_* (a, \ome) |^{2\mu} ( a \Phi_0 + \Theta_* (a, \ome) ) \Big) =0.
\end{equation*}
This is an implicit equation w.r.t. two real parameters $a$ and $\ome$, we assume $a\neq0$ and recast it in the following form
\beq \label{bifeq}
f( a, \ome) \equiv  (\ome - E_0 ) - a^{2\mu} \Bigg( \Phi_0 , \lf|  \Phi_0 + \f{\Theta_* (a, \ome)}a  \ri|^{2\mu} \lf(  \Phi_0 + \f{\Theta_* (a, \ome)}a  \ri) \Bigg) =0.
\eeq
We want to use the  Implicit Function Theorem in \eqref{bifeq} to make explicit $\ome (a)$. By the bound \eqref{thetabdd} it is immediate that $f(a,\omega)$ is continuous, moreover 
\[\partial_\ome f( a ,  \ome) 
=  1- a^{2\mu}  
 \Bigg( \Phi_0 ,   (2\mu+1) \lf|  \Phi_0 + \f{\Theta_* (a, \ome)}a \ri|^{2\mu}  \f{\partial_\ome \Theta_* (a, \ome)  }a \Bigg).
\]
Which shows that  $\partial_\ome f( a ,  \ome)$ is also continuous. Notice that
\begin{align*}
\partial_\ome  \Theta_* (a, \ome) & = - ( D_\Theta G )^{-1}( a,  \Theta_* (a, \ome) , \ome ) \, D_\ome G (a,  \Theta_* (a, \ome) , \ome ) \\
& =  - ( D_\Theta G )^{-1}( a,  \Theta_* (a, \ome) , \ome ) Q  \Theta_* (a, \ome).
\end{align*}
Hence, by \eqref{thetabdd}, we have 
\be 
\vertiii{ \partial_\ome  \Theta_* (a, \ome) }_H \leq C a^{2\mu+1},
\ee
which implies that $\partial_\ome f (0, E_0 ) \neq 0$. We conclude that  \eqref{bifeq} defines uniquely a continuous  function $\ome_* (a)$ in a neighborhood of the origin  such that $\ome_* (0) = E_0$. Moreover, it is clear that for small $a$ 
\be 
a^{2\mu} \lf( \Phi_0 , \lf|  \Phi_0 + \f{\Theta_* (a, \ome)}a  \ri|^{2\mu} \lf(  \Phi_0 + \f{\Theta_* (a, \ome)}a \ri) \ri)  \geq 0
\ee
and then $\ome_* - E_0 \geq 0 $ that is $\ome_* \geq E_0$. We can give a more precise asymptotic behavior, that is
\begin{equation*}
 \ome_*(a) =  E_0 + a^{2\mu} \| \Phi_0\|_{2\mu + 2}^{2\mu+2}  + O(a^{4\mu}).
\end{equation*}
Concerning the regularity properties of $\omega_*$, exploiting the identity $\partial_a f(a,\omega_*(a)) = 0$ we conclude that $\omega_* \in C^1(0,\ve)$. Indeed, by  L'H\^opital's rule,  we infer 
\[ \| \Phi_0\|_{2\mu + 2}^{2\mu+2}  = \lim_{a\to0^+} \frac{\omega_*(a)-E_0}{a^{2\mu}} =  \lim_{a\to0^+} \frac{\omega_*(a)' }{2\mu a^{2\mu-1}},\] hence 
\[ \omega_*(a)' = 2\mu \| \Phi_0\|_{2\mu + 2}^{2\mu+2}   a^{2\mu -1} + o(a^{2\mu -1}),\]
which guarantees that $\omega_*$ is strictly increasing, hence invertible,  in $(0,\ve)$. We denote its inverse by $a_*(\omega)$. Obviously $a_*\in C^1(E_0,E_0+\de)$ and 
\begin{equation}\label{lemon}
a_*(\ome) = \lf( \frac{\ome-E_0}{ \| \Phi_0\|_{2\mu + 2}^{2\mu+2}}  \ri)^{\f{1}{2\mu} } + O \left(( \ome-E_0 )^{\f{1}{2\mu}+1 }\right),
\end{equation}
by the inequality $(A+B)^{\frac1{2\mu} }- B^{\frac1{2\mu}} \leq C A^{\frac{1}{2\mu}-1} B$, which holds true for all $0 < B < A/2$. 
The sought  solution  is given by 
\begin{equation*}
 \Phi(\omega) = a_*(\omega) \Phi_0 +  \Theta_* (a_*(\omega), \ome ).
\end{equation*}
We are left to prove properties  \eqref{mPhiome} and \eqref{EPhiome}. 

As  $\ome-E_0\to 0$, due to  \eqref{thetabdd} and \eqref{lemon} we have 
\begin{equation*}
m(\omega)=\|\Phi(\ome)\|^2  =   a_*(\ome)^2  + \|  \Theta_* (a_*(\omega), \ome )\|^2   =     \lf( \frac{\ome-E_0}{ \| \Phi_0\|_{2\mu + 2}^{2\mu+2}}  \ri)^{\f{1}{\mu} }  + o\left(( \ome-E_0 )^{\f{1}{\mu}}\right),
\end{equation*}
which proves Eq. \eqref{mPhiome}. By the regularity of $a_*$ and $\Theta_*(a,\omega)$ it follows that $m(\omega)$ is in $C^1(E_0,E_0+\de)$ and it is invertible because 
 \[ \left(\| \Phi_0\|_{2\mu + 2}^{2\mu+2}\right)^{-\frac{1}{\mu} } = \lim_{\omega\to E_0^+} \frac{m(\omega)}{(\omega - E_0)^{\frac1\mu}} = \lim_{\omega\to E_0^+} \mu \frac{m(\omega)'}{(\omega - E_0)^{\frac1\mu-1}}. \] 
Denoting its inverse by $\omega (m)$ one has 
\begin{equation}\label{lemon2}
\lf( \frac{ \ome(m)-E_0 }{ \| \Phi_0\|_{2\mu+2 }^{2\mu+2 } } \ri)^{\f{1}{\mu} } = m +o(m).
\end{equation}
Computing the energy we have that 
\[\begin{aligned}
 E[\Phi(\ome)] =& E^{lin} [\Phi(\ome)] - \frac{\|\Phi(\omega)\|_{2\mu+2}^{2\mu+2}}{\mu+1} \\
  =& 
 -E_0 a_*(\omega)^2+ (\Theta_*(a_*(\omega),\omega), H \Theta_*(a_*(\omega),\omega)) - \frac{\|a_*(\omega)\Phi_0 + \Theta_*(a_*(\omega),\omega) \|_{2\mu+2}^{2\mu+2} }{\mu+1}     \\  =& 
 -E_0 a_*(\omega)^2+ o(a_*(\omega)^2).
\end{aligned}
\]
$ E(m) = E[\Phi(\ome(m))]$ is continuous as a function of $m$ (indeed it is $C^1(0,\tilde m)$ for $\tilde m$ small enough), and recalling \eqref{lemon} and \eqref{lemon2} we get \eqref{EPhiome}. 
\end{proof}

\section{Main theorem\label{s:mainth}}
\begin{proof}[Proof of Theorem \ref{t:prob1}]
We prove first that $mE_0<\nu<\infty$. The lower bound $\nu>mE_0$ is a direct consequence of the fact that $E[\Psi] < E^{lin}[\Psi]$, for all $\Psi\in \EE$, and that, by the definition of $E_0$  and $E^{lin}$, one has
\[\inf \{ E^{lin}[\Psi] \text{ s.t. } \Psi\in \EE , \, M[\Psi]=m  \} = -mE_0.\]\\
To prove that $\nu<+\infty$ we first note that, by using  H\"older and   Gagliardo-Nirenberg inequalities, one can prove the bounds:  
\begin{equation}\label{4.1a}
 \|\Psi\|_{{2\mu+2}}^{2\mu+2} \leq c \|\Psi\|_{H^1}^{\mu} \|\Psi\|^{2+\mu}; \end{equation}
\begin{equation}\label{W-1}
(\Psi,W_- \Psi)  \leq \|W_-\|_r \|\Psi\|^2_{2r/(r-1)} \leq c \|W_-\|_r  \|\Psi\|_{H^1}^{2\alpha}\|\Psi\|_q^{2(1-\alpha)} \end{equation}
for all $q\in[2,2r/(r-1)]$ and with $ \alpha= \frac{2}{2+q}\left(1-\frac{q(r-1)}{2r}\right)$; and 
\begin{equation}\label{4.1b}
 |\Psi(\v)|^2\leq \|\Psi\|_{\infty}^2\leq c  \|\Psi\|_{H^1}\|\Psi\| \quad  \forall \v \in V \,.
\end{equation}
We remark that the inequalities \eqref{4.1a} - \eqref{4.1b} hold true for any connected finite graph. 
 If $M[\Psi] = m$, by \eqref{4.1a} - \eqref{4.1b} we have 
\[
 E[\Psi] + m
 \geq \| \Psi \|_{H^1}^2 -C \frac{m^{\frac{2+\mu}{2}}}{\mu+1}\|\Psi\|_{H^1}^{\mu}
  - C \sqrt m \sum_{\underline v\in V_v}|\alpha(\underline v)|   \|\Psi\|_{H^1} -  Cm^{1-1/(2r)} \|W_-\|_r  \|\Psi\|_{H^1}^{1/r}.
\]
We notice that for any $a,b,c,d>0$, $r\geq 1$, and   $0< \mu < 2$ there exist $\de,\beta >0$ such that $a x^2 - bx^\mu  -cx -d x^{1/r} > \de x^2 - \beta $, for any $x\geq0$, then  
\beq
\label{e:dec}
 E[\Psi] + m
 \geq \de \| \Psi\|_{H^1}^2 -\beta \,,
\eeq
which implies $\nu \leq \beta + m$.\\

In the  remaining part of the proof we shall prove that we can choose $m^*$ such that for $m<m^*$ minimizing sequences have a  convergent subsequence.

Let $\{\Psi_n\}_{n\in\NA}$ be a minimizing sequence, i.e., $\Psi_n \in\EE$, $M[\Psi_n]=m$, and $\lim_{n\to\infty} E[\Psi_n] = -\nu$. Concerning the mass constraint, we remark that it is enough to assume $M[\Psi_n]\to m$ as $n\to \infty$, in such a case one can define $\widetilde \Psi_n = \sqrt{m} \Psi_n /\|\Psi_n\|$ and note that $\lim_{n\to\infty} E[\widetilde\Psi_n] = \lim_{n\to\infty} E[\Psi_n] $. 

We shall prove that there exists $\hat \Psi \in H^1 (\GG)$ such that $M[\hat \Psi] = m $, $E[\hat \Psi] =-\nu$ and $\Psi_n \to \hat \Psi$ in $ H^1 (\GG)$.

We can assume that 
$E[\Psi_n] \leq -\nu/2$ then by inequality \eqref{e:dec}, up to taking a subsequence,  we can assume that 
\[
\sup_{n\in\NA}\|\Psi_n\|_{H^1}\leq \infty,\]
moreover the following lower bound holds true
\beq
\label{e:floor}
 \frac{1}{\mu+1} \| \Psi_n\|_{2\mu+2}^{2\mu+2} +(\Psi,W_-\Psi) +  \sum_{\v \in V_- } |\al(v)| |\Psi_{n}(\v)|^2 \geq
\frac{\nu}2 \,.
\eeq

Next we use Lem. \ref{l:cc} and  prove that  vanishing and dichotomy
cannot occur for $\{\Psi_n\}_{n\in\NA}$. Set $\tau = \lim_{t\to\infty}\liminf_{n\to\infty}
\rho(\Psi_n,t)$. First we prove that  vanishing cannot occur. If
$\tau=0$,  then by Lem. \ref{l:cc} there would exist a subsequence $\Psi_{n_k}$
such that $\| \Psi_{n_k}\|_{p} \to 0 $ for all $2<p\leq \infty$ but
this, together with Eqs. \eqref{W-1} and \eqref{4.1b},   would contradict 
\eqref{e:floor}. \\
To prove that dichotomy cannot occur, suppose $0<\tau<m$, then there
would exist $\VV_k$ and $\WW_k$ satisfying \eqref{dic1}-\eqref{dic7}.
In particular we know that
\[
 \liminf_{k\to \infty} \left(\|\Psi_{n_k}' \|^2 - \| \VV_k' \|^2
- \| \WW_k' \|^2 \right)  \geq 0
\]
\[
\lim_{k\to \infty} \lf( \| \Psi_{n_k}\|_p^p - \|\VV_k \|_p^p -
\| \WW_k \|_p^p \ri)=0 \qquad 2\leq p < \infty 
\]
and 
\be
\lim_{k\to \infty}\left||\Psi_{n_k}(\v)|^2-  |\VV_{k}(\v)|^2 -| \WW_{k}(\v)|^2\right| = 0\,.
\ee
Moreover we claim that 
\begin{equation}\label{claim1}
\lim_{k\to\infty}  (\Psi_{n_k}, W\Psi_{n_k}) - (\VV_k,W \VV_k) - (\WW_k, W\WW_k) \geq  0  ,
\end{equation}
we postpone the proof of this claim to the end of the discussion.  Summing up, we arrive at
\be
\liminf_{k\to\infty} \lf(
E[\Psi_{n_k} ] - E[ \VV_k] - E[\WW_k]
\ri) \geq 0 \,,
\ee
which implies
\beq
\label{e:black-1}
\limsup_{k\to\infty} \lf(
 E[ \VV_k] + E[\WW_k]
\ri) \leq -\nu \,.
\eeq
Notice that, given $\Psi\in\EE$ and $\de >0$, then
\[
E[\Psi] = \frac{1}{\de^2} E[\de \Psi] + \frac{\de^{2\mu} -1}{\mu+1} \| \Psi
\|_{2\mu+2}^{2\mu+2}.
\]
We remark that $\VV_k, \WW_k \in \EE$, since $\Psi_{n_k}$ satisfies   the continuity condition at
the vertices  and the multiplication with the cut-off functions
preserves that. Let $\de_k= \sqrt {m / M[\VV_k]}$ and $\ga_k = \sqrt {m / M[\WW_k]}$  such that $M[\de_k \VV_k] ,\,
M[\ga_k \WW_k] =m$. Then, using the above equality and the fact that
$E[\de_k \VV_k], E[\ga_k \WW_k] \geq -\nu$,  one has
\[
E[\VV_k] \geq - \frac{\nu}{\de^2_k} + \frac{\de^{2\mu}_k -1}{\mu+1} \| \VV_k
\|_{2\mu+2}^{2\mu+2}
\]
\[
E[\WW_k] \geq - \frac{\nu}{\ga^2_k} + \frac{\ga^{2\mu}_k -1}{\mu+1} \| \WW_k
\|_{2\mu+2}^{2\mu+2}
\]
from which 
\[
E[\VV_k]+E[\WW_k] \geq -\nu \lf( \frac{1}{\de^2_k} + \frac{1}{\ga^2_k} \ri) +
\frac{\de^{2\mu}_k -1}{\mu+1} \| \VV_k \|_{2\mu+2}^{2\mu+2} +
\frac{\ga^{2\mu}_k -1}{\mu+1} \| \WW_k \|_{2\mu+2}^{2\mu+2}\,.
\]
Notice that  by \eqref{dic4}
\[
\frac{1}{\de^2_k} \to \frac{\tau}{m} \qquad \qquad \frac{1}{\ga^2_k} \to 1-\frac{\tau}{m}\,.
\]
Let $\theta = \min \{ (\tau/m)^{-\mu} , (1-\tau/m)^{-\mu} \}$ and notice that $\theta
>1$ since $0<\tau/m <1$. Therefore
\begin{align}
\label{e:black-2}
\liminf_{k\to\infty} \lf(
 E[ \VV_k] + E[\WW_k]
\ri) 
&\geq -\nu + \frac{\theta -1}{\mu+1} \liminf_{k\to\infty} \| \Psi_{n_k}
\|_{2\mu+2}^{2\mu+2} > -\nu,
\end{align}
where we used the fact that $\liminf_{k\to\infty} \| \Psi_{n_k}
\|_{2\mu+2}^{2\mu+2} \neq 0$. The latter claim is proved by noticing
that $\liminf_{k\to\infty} \| \Psi_{n_k} 
\|_{2\mu+2}^{2\mu+2} = 0$,  together with $\| \Psi_{n_k}
\|_{H^1}$ bounded and Eqs. \eqref{W-1} and \eqref{4.1b},  would imply $\liminf_{k\to\infty} (\Psi_{n_k},W_-\Psi_{n_k}) =0$
 and $\liminf_{k\to\infty} \|\Psi_{n_k}\|_\infty =0$. Hence, there would be a contradiction with   inequality \eqref{e:floor}. We conclude that if $0<\tau<m$  we get a contradiction, cfr. inequalities \eqref{e:black-1} and \eqref{e:black-2}. To end the analysis of the case $0<\tau<m$ we are left to prove the claim \eqref{claim1}.  We rewrite $W = W_+- W_-$ and consider first the term with $W_+$. We have that 
\[\begin{aligned}
 &(\Psi_{n_k}, W_+\Psi_{n_k}) - (\VV_k,W_+ \VV_k) - (\WW_k, W_+\WW_k)     \\ 
=&   \sum_e \int_{I_e}(W_+)_e \left[1 - (\Theta_k)_e^2 - (\Phi_k)_e^2\right] |(\Psi_{n_k})_e|^2 dx \geq 0. 
\end{aligned}\]
Since $\VV_k$ and $\WW_k$ have disjoint supports, we have that 
\begin{equation*}\begin{aligned}
&\left| (\Psi_{n_k}, W_-\Psi_{n_k}) - (\VV_k,W_- \VV_k) - (\WW_k, W_-\WW_k) \right|  \\ 
\leq &  |(\ZZ_k, W_-\ZZ_k)| +2|(\VV_k,W_- \ZZ_k)| +2 |(\WW_k, W_-\ZZ_k) | \\  
\leq & | (\ZZ_k, W_-\ZZ_k)| +2(\VV_k,W_-\VV_k)^{1/2 }(\ZZ_k,W_- \ZZ_k)^{1/2} +2(\WW_kW_-\WW_k)^{1/2 }(\ZZ_k,W_- \ZZ_k)^{1/2} .
\end{aligned}\end{equation*}
The terms containing $\VV_k$ and $\WW_k$ are bounded by Lemma \ref{l:cc} and inequality \eqref{W-1}. The terms containing $\ZZ_k$, go to zero by  inequality \eqref{W-1}  and because $\|\ZZ_k\|\to 0$ by Eq. \eqref{ZZk}. From which the claim \eqref{claim1} follows. 

Since  $0\leq\tau<m$ leads us to a contradiction, it must be $\tau=m$. 

Now we prove that for $m<m^\ast$ the minimizing sequence is not {\em runaway}. Here the limitation on the mass plays a role for the first time.
By absurd suppose that $\{\Psi_n\}_{n\in\NA} $ is {\em runaway}, then we have that 
\begin{equation}\label{limit}
\lim_{n\to\infty} \Psi_{n} (\v) =0\quad \forall \underline{v}\in V\qquad \text{and}\qquad \lim_{n\to\infty}(\Psi_n, W_- \Psi_n)= 0.
\end{equation} The first limit  is a direct consequence of Lem. \ref{l:cc}, Eq. \eqref{e:run-1}. To prove the second one, assume that $\Psi_n$ escapes at infinity on the external  edge $e^*$ (this can always be done up to taking a subsequence). We note that
\begin{equation*}
\lim_{n\to\infty}\int_{I_e} (W_-)_e |(\Psi_n)_e|^2 dx = 0 \qquad \forall e \neq e^* ,
\end{equation*}
this is a  direct consequence of Lemma \ref{l:cc} and inequality  \eqref{W-1} applied to the edge $I_e$. We are left to prove that 
\begin{equation}\label{holiday}
\lim_{n\to \infty}\int_{0}^{+\infty} (W_-)_{e^*} |(\Psi_n)_{e^*}|^2 dx = 0 . 
\end{equation}
 We start by noticing that  $\|\Psi_n\|_{H^1}$ is uniformly bounded, hence, so is $\|\Psi_n\|_p$ for all $p\in [2,+\infty]$, by \eqref{gajardo3} (with $q=2$). As a consequence, we have that  for any $\ve>0$ there exists $R>0$ (independent of $n$) such that
 \begin{equation*}
\int_{R}^{+\infty} (W_-)_{e^*} |(\Psi_n)_{e^*}|^2 dx \leq \|(W_-)_{e^*}\|_{L^r(R,\infty)} \|\Psi_n\|_{2r'}^2 \leq \ve,
\end{equation*}
with $r'$ such that $r^{-1}+{r'}^{-1} =1 $.  For such $R$, there exists $n_0$ such that for all $n>n_0$ one has 
\begin{equation*}
\int_{0}^{R} (W_-)_{e^*} |(\Psi_n)_{e^*}|^2 dx \leq  \|W_-\|_r \|(\Psi_n)_{e^*}\|_{L^{2r'}(0,R)}^2 \leq \ve  
\end{equation*}
by \eqref{e:run-1} (see also Rem. \ref{r:2.8}), from which the second limit in \eqref{limit}. 

Recalling that, by Lem. \ref{l:cc} - Eq. \eqref{e:run-1}, one has $\lim_{n\to\infty}\|(\Psi_n)_e\|_{L^{2\mu+2}(I_e)} =0$ for all $e\neq e^*$, and by Eq. \eqref{limit}, we infer 
 \begin{equation}
 \label{little}
 \lim_{n\to \infty} E[\Psi_n]  \geq \lim_{n\to\infty }  \int_0^\infty |(\Psi_n)_{e^*}'|^2 dx -\frac{1}{\mu+1} \int_0^\infty  |(\Psi_n)_{e^*}|^{2\mu+2} dx.
 \end{equation}
 Let  $\chi:\RE_+ \to [0,1]$ be a  function such that $\chi \in C^\infty(\RE_+)$, $\chi(0) = 0$ and $\chi(x)=1$ for all $x\geq 1$.  Define 
\[\psi_n^*(x) :=  \chi(x)(\Psi_n)_{e^*}(x) , \]
so that $\psi_n^*(0)= 0$, and ${\|\psi_n^*}'\|_{L^2(\RE_+)}^2 \leq c$. By Lem. \ref{l:cc} - Eq. \eqref{e:run-1}, for all $p\geq 2$,
\begin{equation}\label{psinstar1}
\lim_{n\to\infty} \|\Psi_n\|_p^p = \lim_{n\to\infty} \|(\Psi_n)_{e^*}\|_{L^p((0,\infty))}^p = \lim_{n\to\infty} \|\psi_n^*\|_{L^p((0,\infty))}^p,
\end{equation}
where we used the fact that $\lim_{n\to\infty} \|(\Psi_n)_{e^*}\|_{L^p((0,1))} = 0  $, and the trivial bound $ \|\psi_n^*\|_{L^p((0,1))} \leq  \|\chi\|_{L^\infty((0,1))} \|(\Psi_n)_{e^*}\|_{L^p((0,1))}$.  In particular, $ \lim_{n\to\infty} \|(\Psi_n)_{e^*}\|_{L^2((0,\infty))}^2 = \lim_{n\to\infty} \|\psi_n^*\|_{L^2((0,\infty))}^2 = m$. 
Moreover we have that 
 \begin{equation}\label{psinstar2}
  \lim_{n\to\infty } \frac12 \int_0^\infty |(\Psi_n)_{e^*}'|^2 dx \geq   \lim_{n\to\infty } \frac12 \int_0^\infty |{\psi_n^*}'|^2 dx. 
\end{equation}
 To prove the latter inequality, we note that 
\[\begin{aligned}
 \lim_{n\to\infty } \int_0^\infty |(\Psi_n)_{e^*}'|^2 - |{\psi_n^*}'|^2 dx 
 =&   
 \lim_{n\to\infty } \int_0^\infty |(\Psi_n)_{e^*}'|^2\left(1-\chi^{2}\right)  dx  \\ 
 & +
 \lim_{n\to\infty } \int_0^1 |(\Psi_n)_{e^*}|^2 {\chi'}^{2} + 2\chi {\chi'} \Re \overline {(\Psi_n)_{e^*}'}(\Psi_n)_{e^*} dx \\
 = &   \lim_{n\to\infty } \int_0^1 |(\Psi_n)_{e^*}'|^2\left(1- \chi^{2}\right)  dx  \geq 0 ,
\end{aligned}\]
where we used again Lem. \ref{l:cc} - Eq. \eqref{e:run-1} and the bounds $\|\chi\|_\infty, \|\chi'\|_\infty\leq c $. 

We have the following chain of inequalities/identities  
\begin{align}
&\lim_{n\to \infty} E[\Psi_n] \nonumber \\
& \geq 
\lim_{n\to\infty } \int_0^\infty |{\psi_n^*}'(x)|^2 dx -\frac{1}{\mu+1} \int_0^\infty  |\psi_n^*(x)|^{2\mu+2} dx \nonumber\\
&\text{(we used Eqs. \eqref{little}, \eqref{psinstar1} and \eqref{psinstar2})} \nonumber \\
&   \geq \inf \Big\{ \int_0^\infty |\psi ' (x)|^{2} dx -\frac{1}{\mu+1} \int_0^\infty |\psi (x)|^{2\mu+2} \, dx  \text{ s.t. } \psi\in H^1(\RE^+), \, \psi(0)=0\, ,\| \psi\|_{L^2(\RE^+)}^2 =m \Big\} \nonumber \\ 
&\text{(we used the fact that $\psi_n^*\in H^1(\RE_+)$,  $\psi_n^*(0)= 0$, and $\|\psi_n^*\|_{L^2(\RE_+)}^2 \to  m$ as $n\to\infty$)} \nonumber
\\ 
&  =  \inf \Big\{  \int_\RE |{\psi}' (x)|^{2} dx-\frac{1}{\mu+1} \int_\RE|{\psi} (x)|^{2\mu+2} \, dx   \text{ s.t. } \psi\in H^1(\RE), \, \psi(x)=0 \; \forall x\leq 0 ,\| \psi\|_{L^2(\RE)}^2 =m  \Big\}\nonumber \\
&\text{(where we used the fact that $\psi\in H^1(\RE_+)$ and $\psi(0)= 0$ if and only if its zero extension} \nonumber\\
&\text{belongs to $H^1(\RE)$, see, e.g., \cite[Th. 5.29]{AF03}})\nonumber
\\ 
& \geq \inf \left\{  \int_\RE |\psi  ' (x)|^{2} dx-\frac{1}{\mu+1} \int_\RE|\psi (x)|^{2\mu+2} \, dx  \text{ s.t. } \psi\in H^1(\RE), \, \| \psi\|_{L^2(\RE)}^2 =m  \right\} \label{infsol}\\
&\text{(we enlarged the set on which the $\inf$ is taken).} \nonumber
\end{align}

It is well known that the infimum in the latter minimization problem is indeed  attained and that the minimizing function (up to translations and phase multiplications) is given by the soliton profile 
\begin{equation*}
\phi(x) = [ (\mu + 1) \ome_{\RE}]^{\frac{1}{2\mu}} \sech^{\frac{1}{\mu}} (\mu \sqrt{\ome_{\RE}} x).
\end{equation*}
The frequency   $\omega_\RE$ is fixed by the mass constraint through the relation  
\begin{equation*}
m =\|\phi\|^2_{L^2(\RE)} =   2\f{(\mu+1)^{\f 1 \mu }  }{\mu} \ome_\RE^{ \f 1 \mu - \f 1 2} \int_0^1 (1-t^2)^{\f 1 \mu -1} dt,
\end{equation*}
which gives 
\begin{equation*}
\omega_\RE = \left( 2\f{(\mu+1)^{\f 1 \mu }  }{\mu}  \int_0^1 (1-t^2)^{\f 1 \mu -1} dt\right)^{-\frac{2\mu}{2-\mu}} m^{\frac{2\mu}{2-\mu}}.
\end{equation*}
The infimum in the minimization problem \eqref{infsol} is given by  the nonlinear energy  of the soliton 
\begin{equation*}
  	  \int_\RE |\phi'(x) |^{2} dx-\frac{1}{\mu+1} \int_\RE|\phi(x)|^{2\mu+2} \, dx  =
- \f{2-\mu}{2+\mu}  \; \ome_\erre \, m = -\gamma_\mu m^{1+\frac{2\mu}{2-\mu}}, 
\end{equation*}
with $\gamma_\mu= \f{2-\mu}{2+\mu}\left( 2\f{(\mu+1)^{\f 1 \mu }  }{\mu}  \int_0^1 (1-t^2)^{\f 1 \mu -1} dt\right)^{-\frac{2\mu}{2-\mu}}$. So that by the inequality \eqref{infsol}, we conclude that  if $\Psi_n$ is a runaway sequence  it must be 
\begin{equation}\label{lower}
\lim_{n\to \infty} E[\Psi_n] \geq  -\gamma_\mu m^{1+\frac{2\mu}{2-\mu}} .
\end{equation}

To show that for $m$ small enough a minimizing sequence cannot be runaway we  compute the energy   on a trial function.  As trial function we choose the  function $\Phi(\ome)$, with $\omega = \omega(m)$,  given in Th. \ref{t:bif}. By the same theorem we have that the energy  $ E[\Phi(\ome)] = -E_0 m + o(m)$, and by a simple continuity argument we infer that there exists $m^*$  such that $ E[\Phi(\ome)]  < -\gamma_\mu m^{1+\frac{2\mu}{2-\mu}}$ for all  $0<m<m^*$. This, together with the lower bound \eqref{lower},   imply that a minimizing sequence cannot be runaway. 

By Lem. \ref{l:cc} we conclude that for all $0<m<m^*$ there exists a state $\hat\Psi \in\EE$ such that  minimizing sequences converge, up to taking  subsequences, to $\hat\Psi $ in $L^p$ for $p  \geq 2$. In particular, $M[\hat \Psi] = m$, and the potential, vertices, and nonlinear terms in $E[\Psi_n]$ converge to the corresponding ones in $E[\hat\Psi]$. Taking into account also the weak lower continuity of the $H^1$ norm we have
\[
E[ \hat \Psi] \leq \lim_{n\to \infty} E[ \Psi_n] = -\nu
\]
which implies that $E[ \hat \Psi] = -\nu$. Since $E[ \hat \Psi] = \lim_{n\to \infty}  E[ \Psi_n]$ then $\| \hat \Psi ' \| = \lim_{n\to \infty} \|  \Psi_n ' \|$ and we
have proved that $\Phi_n \to \hat \Psi$ in $H^1$.
\end{proof}

\subsection*{Acknowledgments.}  
The authors are grateful to Gregory Berkolaiko, Pavel Exner and Delio Mugnolo for useful discussions. 
D.F. and D.N.  acknowledge the support of FIRB 2012 project ``Dispersive dynamics: Fourier Analysis and Variational Methods'', Ministry of University and
Research of Italian Republic  (code RBFR12MXPO).
C.C. acknowledges the support of the FIR 2013 project ``Condensed Matter in Mathematical Physics'', Ministry of University and
Research of Italian Republic  (code RBFR13WAET)


\begin{thebibliography}{99}

\bibitem{[ACFN1]}
R.~Adami, C.~Cacciapuoti, D.~Finco, and D.~Noja, \emph{Fast solitons on star
  graphs}, Rev. Math. Phys \textbf{23} (2011), no.~4, 409--451.

\bibitem{[ACFN2]}
R.~Adami, C.~Cacciapuoti, D.~Finco, and D.~Noja, \emph{On the structure of
 critical energy levels for the cubic focusing {NLS} on star graphs}, J. Phys.
 A: Math. Theor. \textbf{45} (2012), 192001, 7pp.
 
\bibitem{[ACFN4]}
R.~Adami, C.~Cacciapuoti, D.~Finco, and D.~Noja, \emph{Stationary states of
{NLS} on star graphs}, EPL \textbf{100} (2012), 10003, 6pp.

\bibitem{[ACFN3]}
R.~Adami, C.~Cacciapuoti, D.~Finco, and D.~Noja, \emph{Variational properties
  and orbital stability of standing waves for {NLS} equation on a star graph},
J. Differ.  Equations {\bf 257} (2014), 3738--3777.

\bibitem{acfn-aihp} R. Adami, C. Cacciapuoti, D. Finco, D. Noja,
{\em Constrained energy minimization and orbital stability for the NLS equation on a star graph},  
Ann. Inst. Poincar\'e, An. Non Lin. \textbf{31} (2014), no. 6, 1289--1310. 
 
\bibitem{ACFN16} R.~Adami, C.~Cacciapuoti, D.~Finco, and D.~Noja, \emph{Stable standing waves for a NLS on star graphs as local minimizers of the constrained energy},
J. Differ. Equations \textbf{260} (2016) 7397--7415.

\bibitem{ANcmp} R. Adami, and D.   Noja,
 \emph{Stability and Symmetry-Breaking Bifurcation for the Ground States of a NLS
 with a $\delta'$ Interaction}, Comm. Math. Phys. {\bf 318} (2013), 247--289. 
 
\bibitem{ANV12}
R.~Adami, D.~Noja, and N.~Visciglia \emph{Constrained energy minimization and ground states for NLS with point defects}, Discrete and Continuous Dynamical Systems B
{\bf 18} (2013), 1155--1188. 

\bibitem{AST1}  R. Adami, E. Serra, P.  Tilli, \emph{NLS ground states on graphs}, Calc. Var. and PDEs \textbf{54} (2015), no. 1, 743--761.

\bibitem{AST2} R. Adami, E. Serra, P.  Tilli, \emph{Threshold phenomena and existence results for NLS ground states on metric graphs}, J. Func. An. \textbf{271} (2016), no. 1, 201--223.

\bibitem{AF03} R.~A. Adams,  and J.~J.~F.~ Fournier, Sobolev spaces, Pure and Applied Mathematics Series Vol. 140, Academic press, 2003.

\bibitem{AMN15} F. Ali Mehmeti, K. Ammari, and S. Nicaise, \emph{Dispersive effects for the Schr\"odinger equation on a tadpole graph}, arXiv:1512.05269 [math-ph] (2015).

\bibitem{[BI1]}V. Banica, and L. I. Ignat, \emph{Dispersion for the Schr\"odinger equation on networks},  J. Math. Phys. \textbf{52} (2011), no. 8, 083703, 14pp.

\bibitem{[BI2]}V. Banica, and  L. I. Ignat, \emph{Dispersion for the Schr\"odinger equation on the line with multiple Dirac delta potentials and on delta trees},  Analysis \& P.D.E. {\bf 7} (2014), no. 4, 903--927. 

\bibitem{BerKu} G. Berkolaiko, P. Kuchment, Introduction to Quantum Graphs, Mathematical Surveys and Monographs 186,  AMS (2013).
    
\bibitem{Berko16} G. Berkolaiko, and  W. Liu, \emph{Simplicity of eigenvalues and non-vanishing of eigenfunctions of a quantum graph}, arXiv:1601.06225v2 [math-ph] (2016), to appear on J. Math. Anal. Appl.. 

\bibitem{cfn15} C. Cacciapuoti, D. Finco, D. Noja, \emph{Topology induced bifurcations for the NLS on the tadpole graph}, Phys. Rev. E
\textbf{91} (2015), no. 1,  013206, 8 pp.

\bibitem{Caz03}
T.~Cazenave, {S}emilinear {S}chr\"{o}dinger {E}quations, Courant Lecture Notes in Mathematics, AMS, vol. 10, Providence, 2003.

\bibitem{Caz06}
T.~Cazenave, An introduction to semilinear elliptic equations, Editora
  do IM-UFRJ, Rio de Janeiro, 2006.

\bibitem{[CL]}
T.~Cazenave, and P.-L. Lions, \emph{Orbital stability of standing waves for some
  nonlinear {S}chr\"odinger equations}, Commun. Math. Phys. \textbf{85} (1982),
  549--561.
  
\bibitem{EJ} P. Exner, and M. Jex, \emph{On the ground state of quantum graphs with attractive $\delta$-coupling}, Physics Letters A \textbf{376} (2012), 713--717.

\bibitem{H} S. Haeseler, \emph{Heat kernel estimates and related inequalities on metric graphs}, arXiv:1101.3010v1 [math-ph] (2011).

\bibitem{Kelleretal} M. Keller, D. Lenz, and  R. Wojciechowski, \emph{Note on basic features of large time behaviour of heat kernels}, J. reine angew. Math. \textbf{708} (2015), 73--95.

\bibitem{kirr-kevrekidis-pelinovsky:11} E. Kirr, P. G. Kevrekidis, and D. E. Pelinovsky, \emph{Symmetry-Breaking Bifurcation in the Nonlinear Schr\"odinger Equation with Symmetric Potentials}, Comm. Math. Phys, {\bf 308} (2011), 795--844. 

\bibitem {KPS} V. Kostrykin, J. Potthoff, and R. Schrader, \emph{Contraction Semigroups on Metric Graphs}, 
Proceedings of Symposia in Pure Mathematics \textbf{77} (2008), 423--458.

\bibitem{[KS99]} V.~Kostrykin and R.~Schrader, \emph{{K}irchhoff's rule for quantum wires}, J.Phys. A: Math. Gen. \textbf{32} (1999), no.~4, 595--630.

\bibitem{LL01}
E.~H. Lieb and M.~Loss, Analysis, second ed., Graduate Studies in
  Mathematics, vol.~14, American Mathematical Society, Providence, RI, 2001.

\bibitem{[MP16]}  J. Marzuola, and D. E. Pelinovsky, \emph{Ground states on the dumbbell graph}, Applied Mathematics Research Express 2016, 98--145 (2016).
  
\bibitem{Mugnolo} D. Mugnolo, Semigroup Methods for Evolution Equations on Networks, Springer (2014).

\bibitem{Noja14} D. Noja, \emph{Nonlinear Schr\"odinger equations on graphs:
recent results and open problems}, Phil. Trans. Roy Soc. A {\bf 372} (2014), 20130002,
20 pages.

\bibitem{[NPS15]} D. Noja, D. Pelinovsky, and G. Shaikhova, \emph{Bifurcation and stability of standing waves in the nonlinear Schr\"odinger
equation on the tadpole graph}, Nonlinearity {\bf 28} (2015), 2343--2378.

\bibitem{[PS16]} D. E. Pelinovsky, and G. Schneider, \emph{Bifurcations of standing localized waves on periodic graphs},  arXiv:1603.05463v1 [math.DS] (2016).

\bibitem{RSIV} M. Reed, and B. Simon, Methods of modern mathematical physics IV, Analysis of operators, Academic Press, London (1978).

\end{thebibliography}
\end{document}